\documentclass[12pt, a4paper]{article}
\usepackage{arxiv}
\usepackage[utf8]{inputenc} 
\usepackage[T1]{fontenc}    

\usepackage{amsmath}
\usepackage{amsthm}
\usepackage{amsfonts}
\usepackage[dvipsnames]{xcolor}
\usepackage{graphicx}
\usepackage{hyperref}
\usepackage[capitalise,nameinlink]{cleveref}

\hypersetup{
	colorlinks = true,
	linkcolor = OliveGreen,
	anchorcolor = black,
	citecolor = RedViolet,
	filecolor = cyan,
	menucolor = red,
	runcolor = cyan,
	urlcolor = MidnightBlue,
}

\usepackage{mathtools} 
\usepackage{arydshln}  
\usepackage{enumitem}
\usepackage{soul}
\setlist[itemize]{leftmargin=1cm}

\newtheorem{theorem}{Theorem}
\newtheorem{lemma}{Lemma}
\newtheorem{definition}{Definition}
\newtheorem{remark}{Remark}

\newcommand{\dd}{\mathrm{d}}

\title{Multilevel Delayed Acceptance MCMC}

\author{%
	Mikkel B. Lykkegaard\thanks{Corresponding author.} \\
	Centre for Water Systems and \\
	Institute for Data Science and AI\\
	University of Exeter \\
	EX4 4QF, United Kingdom \\
	\texttt{m.lykkegaard@exeter.ac.uk} \\
	\And
	Tim J. Dodwell \\
	The Alan Turing Institute and \\
	Institute for Data Science and AI \\
	University of Exeter \\
	EX4 4QF, United Kingdom \\
	\texttt{t.dodwell@exeter.ac.uk} \\
	\And
	Colin Fox \\
	Department of Physics \\
	University of Otago \\
	Dunedin 9016, New Zealand \\
	\texttt{colin.fox@otago.ac.nz} \\
	\And
	Grigorios Mingas \\
	The Alan Turing Institute  \\
	British Library, 96 Euston Road \\
	NW1 2DB, United Kingdom\\
	\texttt{gmingas@turing.ac.uk} \\
	\And
	Robert Scheichl \\
	Institute for Applied Mathematics and \\
	Interdisciplinary Center for Scientific Computing \\ 
	Heidelberg University \\
	69120 Heidelberg, Germany \\
	\texttt{r.scheichl@uni-heidelberg.de} \\
}

\begin{document}

\maketitle

\begin{abstract}
We develop a novel Markov chain Monte Carlo (MCMC) method that exploits a hierarchy of models of increasing complexity to efficiently generate samples from an unnormalized target distribution. Broadly, the method rewrites the Multilevel MCMC approach of Dodwell \textit{et al.} (2015) in terms of the Delayed Acceptance (DA) MCMC of Christen \& Fox (2005). In particular, DA is extended to use a hierarchy of models of arbitrary depth, and allow subchains of arbitrary length. We show that the algorithm satisfies detailed balance, hence is ergodic for the target distribution. Furthermore, multilevel variance reduction is derived that exploits the multiple levels and subchains,  and an adaptive multilevel correction to coarse-level biases is developed. Three numerical examples of Bayesian inverse problems are presented that demonstrate the advantages of these novel methods. The software and examples are available in \texttt{PyMC3}.
\end{abstract}


\section{Introduction}

Sampling from an unnormalised posterior distribution $\pi(\cdot)$ using Markov Chain Monte Carlo (MCMC) methods is a central task in computational statistics. This can be a particularly challenging problem when the evaluation of $\pi(\cdot)$ is computationally expensive and the parameters $\theta$ and/or data ${\bf d}$ defining $\pi(\cdot)$ are high-dimensional. The sequential (highly) correlated nature of a Markov chain and the slow converge rates of MCMC sampling, means that often many MCMC samples are required to obtain a sufficient representation of a posterior distribution $\pi(\cdot)$. Examples of such challenging problems frequently occur in Bayesian inverse problems, image reconstruction and probabilistic machine learning, where simulations of the measurements (required to calculate a likelihood function) depend on the evaluation of complex mathematical models (e.g. a system of partial differential equations) or the evaluation of prohibitively large data sets.

The topic of MCMC methods is a rich and active field of research. While the basic idea of the original Metropolis--Hastings algorithm \cite{Met53, Has70} is almost embarrassingly simple, it has given rise to a wide variety of algorithms tailored to different applications. Most notably, the Gibbs sampler \cite{geman_stochastic_1984}, which samples each variable conditional on the other variables, the Metropolis Adjusted Langevin Algorithm (MALA, \cite{roberts_exponential_1996, roberts_optimal_1998}), Hamiltonian Monte Carlo (HMC, \cite{duane_hybrid_1987}) and the No-U-Turn Sampler (NUTS, \cite{nuts}), which all exploit gradient information to improve the MCMC proposals. We would also like to highlight the seminal work of Haario et al. \cite{haario_adaptive_2001} on the Adaptive Metropolis sampler that launched a new paradigm of adaptive MCMC algorithms (see e.g. \cite{atchade_adaptive_2006, andrieu_tutorial_2008, roberts_examples_2009, vrugt_accelerating_2009, zhou_hybrid_2017, cui_posteriori_2019}).

The most efficient MCMC methods cheaply generate candidate proposals, which have a high probability of being accepted, whilst being almost independent from the previous sample. In this paper, we define a MCMC approach capable of accelerating existing sampling methods, where a hierarchy (or sequence) $\pi_0(\cdot), \ldots, \pi_{L-1}(\cdot)$ of computationally cheaper approximations to the exact posterior density $\pi(\cdot) \equiv \pi_L(\cdot)$ are available. As with the original Delayed Acceptance (DA) algorithm, proposed by Christen and Fox \cite{Chr05}, short runs of MCMC subchains, generated using a computationally cheaper, approximate density $\pi_{\ell-1}(\cdot)$, are used to generate proposals for the Markov chain targeting $\pi_\ell(\cdot)$. The original DA method formulated the approach for just two levels and a single step on the coarse level. In this paper we extend the method by recursively applying DA across a hierarchy of model approximations for an arbitrary number of steps on the coarse levels -- a method we term {\em Multilevel Delayed Acceptance} (MLDA). There are clear similarities with Multilevel Monte Carlo sampling methods, first proposed by Heinrich \cite{heinrich_multilevel_2001} and later by Giles \cite{giles_multilevel_2008}, which have been widely studied for forward uncertainty propagation problems (see e.g. \cite{cliffe_multilevel_2011, barth_multilevel_2011, charrier_finite_2013, teckentrup_further_2012}) and importantly have been extended to Bayesian inverse problems in the Multilevel Markov Chain Monte Carlo (MLMCMC) approach by Dodwell {\em et al.}~\cite{dodwell_hierarchical_2015} as well as to the Multi-Index setting \cite{haji-ali_multi-index_2016,jasra_multi-index_2017}. 

The fundamental idea of multilevel methods is simple: We let the cheaper (or \textit{coarse}) model(s) do most of the work. In the context of sampling, be it Monte Carlo or MCMC, this entails drawing more samples on the coarser levels than on the finer, and use the entirety of samples across all model levels to improve our Monte Carlo estimates. Additionally, in the context of MCMC, the samplers on the coarser levels inform the samplers on the finer levels by filtering out poor MCMC proposals, effectively boosting the acceptance rate and hence computational efficiency on the finer levels.

The multilevel MCMC algorithm of Dodwell {\em et al.}~\cite{dodwell_hierarchical_2015} achieves these goals, and, importantly, provides a multi-level estimator for quantities of interest, utilising the output of all chains, to allow tuning of work at each level to maximise variance reduction per compute effort. MLMCMC also allows parallelisation across levels by running chains at the coarser levels independently of the finer. However, a by-product of the latter property is that MLMCMC only produces provably unbiased estimates in the limit of infinitely long coarse chains; see \cref{sec:mlmcmc_comparison}. This is a potential problem as computational efficiency may require quite short coarse chains.

One of our main motivations for reworking MLMCMC was to develop a MCMC that could operate with multiple levels of approximation for which we can write a multi-level estimator, 
and that is provably unbiased for finite-length coarse chains. This paper reports the resulting algorithm, that extends the delayed-acceptance MCMC to a multi-level setting with finite-length coarse chains. Those extensions pose several challenges requiring novel solutions;  (1)  As mentioned above, DA evaluates proposals using a single step on the coarse level; the extension to finite-length subchains is presented in \cref{sec:rst}. (2) A less-obvious challenge is that MLMCMC operates with a different state variable at each level, with fewer components in the state at coarser levels, whereas DA uses the same state at both levels; Extension of DA to using embedded state spaces is presented in \cref{sec:embedding}, where the extra `modes' at the fine level are proposed using an additional kernel. The extension to a multi-level DA is then straightforward by recursion on levels, as presented in \cref{sec:MLDA}. 
(3) A further challenge is deriving a multi-level estimator for MLDA since the coarse chains in MLDA do not converge to known approximate posterior distributions, unlike MLMCMC where the independence of chains means that, after burn-in, each chain samples from a known approximate distribution. In contrast, the short coarse chains in MLDA are, in a sense, always in burn-in. We overcome this difficulty by randomising subchain length for proposals, as shown in \cref{sec:rst}, and using a fixed subchain length for fine-level estimates to ensure that estimates of equivalent terms in the telescoping sums converge to the same value. That multi-level estimator is presented in \cref{sec:VarianceReduction}. The \emph{adaptive} DA algorithm introduced in~\cite{cui_bayesian_2011} increases significantly the statistical efficiency by constructing 
\emph{a posteriori} error models that improve the approximate posterior distributions at coarse levels; see~\cite{cui_posteriori_2019, fox2020randomized}.  Adaptive error models for MLDA are presented in \cref{sec:adapt}.

Finally, a further challenge is that DA MCMC is inherently sequential and fine-level proposals must be evaluated on the coarse levels, which precludes parallelisation across levels. Whether MLDA can be effectively parallelised remains an outstanding question, that we discuss in \cref{sec:conclusions}.

The paper is structured as follows: In the following section we present the MLDA algorithm, proving detailed balance of each extension of DA. In this process, we develop two additional algorithms, namely \textit{Randomised-Length-Subchain Surrogate Transition} (RST) in \cref{sec:rst} and \textit{Two Level Delayed Acceptance} (TLDA)  in \cref{sec:embedding}, each of which are valid MCMC samplers in their own respect. Throughout these sections we develop algorithms for two levels only, denoted C for `coarse' (the approximate chain) and F for `fine' (the exact chain). In \cref{sec:embedding} we introduce different states at coarse and fine levels, also denoted (with a slight abuse of notation) by subscripts C and F, respectively. A recursive, multilevel DA algorithm is defined in \cref{sec:MLDA} with detailed balance following from previous sections. A comparison of MLDA and MLMCMC is presented in \cref{sec:mlmcmc_comparison} to provide some intuition on similarities and differences of the two algorithms.  MLDA then provides a provably convergent multi-level algorithm for which we develop a multi-level estimator in \cref{sec:VarianceReduction} that can be exploited for variance reduction. Adaptive error models are developed in \cref{sec:adapt}. In \cref{sec:examp}, we demonstrate the algorithm using three examples of Bayesian inverse problems. First, we show that extended subchains on the coarse level can significantly increase the effective sample size 
compared to an equivalent single-level sampler on the fine level, using an example from gravitational surveying. Second, we demonstrate multilevel variance reduction on a predator-prey model, where coarse models are constructed by restricting the length of the time window over which the differential equation model is fitted to data. Third, we demonstrate the multilevel error model in the context of a subsurface flow problem. We show that when we utilize the error model, we can achieve high effective sample sizes on the finest level, even when a very crude approximation is employed as the coarsest model. Conclusions and future work are discussed in \cref{sec:conclusions}.

\section{Multilevel Delayed Acceptance}

In this section we first outline the theoretical foundations of vanilla Metropolis--Hastings based MCMC \cite{Met53,Has70} and the Delayed Acceptance (DA) method proposed by Christen and Fox \cite{Chr05}. We extend DA in two ways: horizontally, by allowing the coarse sampler to construct subchains of multiple coarse samples before proposing a sample on the fine level; 
and vertically, by recursively using DA on an entire hierarchy of models with increasing resolution/accuracy. This constitutes the Multilevel Delayed Acceptance (MLDA) sampler. From this foundation we further develop a multilevel estimator to exploit variance reduction, and a multilevel adaptive error model which improves the statistical efficiency of the algorithm.

\subsection{Basic MCMC, Ergodic Theorems and Delayed Acceptance} \label{sec:basic_mcmc}

To show that MLDA correctly generates samples from the unnormalised
target density $\pi\left(\cdot\right)$ we will build on standard
ergodicity results for Markov chains (see \cite{roberts_general_2004} and references therein).
Each algorithm considered here defines a stochastic iteration on a well-defined state, so defines a Markov chain. Hence, we can apply classical ergodic theorems for Markov chains.

The ergodic theorems for Markov chains (see \cite{roberts_general_2004} and references therein)
state that the chain is $\pi$-ergodic if the chain is $\pi$-irreducible,
aperiodic, and reversible with respect to $\pi$. Essentially, irreducibility
and aperiodicity guarantee that the Markov chain has a unique equilibrium
distribution, while reversible with respect to $\pi$ ensures that $\pi$
is that unique distribution. The condition of $\pi$-irreducibility
is satisfied when the proposal distribution is chosen 
such that the standard Metropolis--Hasting algorithm is $\pi$-irreducible. For algorithms based on delayed acceptance, it is also necessary that the coarse-level approximation is chosen to maintain irreducibility; see \cite[Thm.~1]{Chr05} for precise
conditions on the approximation. Aperiodicity is a mild condition
that is satisfied by any Metropolis--Hastings algorithm with a non-zero
probability of rejection on any $\pi$-positive set; again see
\cite[Thm.~1]{Chr05}. We will assume that the proposal and approximations are chosen so that these conditions hold. Accordingly, we focus on establishing reversibility of algorithms, 
which is equivalent to the stochastic iteration being in detailed balance with the target density
$\pi$; see \cite{Liu08}.

\subsubsection{Metropolis--Hastings MCMC}
Consider first the plain vanilla Metropolis--Hastings algorithm for
sampling from target density $\pi_\text{t}$. Given an initial state
$\theta^{0}$ and a proposal distribution with density function $q\left(\cdot|\theta\right)$,
the Metropolis--Hastings algorithm for generating a chain of length $N$ is given in Alg.~1.

\begin{center}\vspace{0.5cm}
\fbox{%
  \parbox{0.975\textwidth}{\textbf{Algorithm 1. Metropolis--Hastings} (\textbf{MH})\\[1ex] 
\textbf{function}: $ 
                               \left[\theta^{1},\ldots, \theta^{N}\right] 
                        =\textbf{MH}\left(\pi_\text{t}(\cdot),q(\cdot|\cdot),\theta^{0},N \right)  $ \\[1ex]
\begin{tabular}{@{}l@{}p{0.88\textwidth}}
    \textbf{input: }  & density of target distribution $\pi_\text{t}(\cdot)$, density of proposal distribution $q(\cdot|\cdot)$, initial state $\theta^{0}$, number of steps $N$\\
    \textbf{output: } & ordered list of states $\left[\theta^{1},\ldots, \theta^{N}\right]$ \big(or just the final state $\theta^{N}$\big)
    \end{tabular} \\[1ex]
\textbf{for} $j = 0$ to $N-1$ : 
	\begin{itemize}
		\item Given $\theta^{j}$, generate a proposal $\psi$ distributed as $q(\psi|\theta^{j})$,
                
		\item Accept proposal $\psi$ as the next state, i.e. set $\theta^{j+1} = \psi$, with probability
		\begin{equation}
                 \alpha(\psi|\theta^{j}) = \min \left\{1, \frac{\pi_\text{t}(\psi)q(\theta^{j}|\psi)}{\pi_\text{t}(\theta^{j})q(\psi|\theta^{j})} \right\}
                 \label{eq:alphanorm}
         \end{equation}
         otherwise reject $\psi$ and set $\theta^{j+1} = \theta^{j}$. 
	\end{itemize}
  }
}\vspace{0.5cm}
\end{center}

For each $j$, Alg.~1 simulates a fixed stochastic iteration with
$\theta^{j+1}$ being conditionally dependent only on $\theta^{j}$, the state at step
$j$, which can be represented by a fixed (stationary) transition kernel $K\left(y|x\right)$ that 
generates a (homogeneous) Markov chain. For target
density $\pi_\text{t}$, detailed balance may be written
\[
\pi_\text{t}\left(x\right)K\left(y|x\right)=\pi_\text{t}\left(y\right)K\left(x|y\right),
\]
which, in general, is the property that $K$ is self-adjoint in the measure
$\pi_\text{t}$. See \cite[Sec. 5.3]{Liu08} for a nice method for showing that $K$
simulated by \textbf{MH} Alg. 1 is in detailed balance with $\pi_\text{t}$, and also for a more
general class of acceptance probabilities. 

Hence, under mild conditions on the proposal density $q$ and the initial state $\theta^{0}$, 
the ergodic theorem for Markov chains applies, which guarantees that the $j$-step
density converges to $\pi_\text{t}$, asymptotically as $j\to \infty$. Hence, the Markov chain is $\pi_\text{t}$-ergodic.

A common choice of proposal distributions for inverse problems in multiple dimensions are random-walk proposals, though these typically lead to adjacent states of the chain being highly correlated, resulting in high computational cost to estimate posterior expectations with a desired accuracy. In the following we do not discuss the choice of proposal $q$, though in some sense our primary concern is how to improve a proposal once chosen. We also do not discuss the choice of initial state.

The following lemma gives an alternative form of the acceptance probability in 
Eq.~\cref{eq:alphanorm} used later.
\begin{lemma} \label{lem:cf2}
 If the proposal transition kernel  $q(\cdot|\cdot)$ in Alg.~1 is in detailed balance with some distribution $\pi^*$, then the acceptance probability \cref{eq:alphanorm} may be written
  		\begin{equation}
                 \alpha(\psi|\theta^{j}) = \min \left\{1, \frac{\pi_\text{t}(\psi)\pi^*(\theta^{j}))}{\pi_\text{t}(\theta^{j})\pi^*(\psi)} \right\}
                 \label{eq:alphanormstar}
         \end{equation}
\end{lemma}
\begin{proof}
 Substitute the detailed balance statement $\pi^*(\psi)q(\theta^{j}|\psi) = \pi^*(\theta^{j}))q(\psi|\theta^{j})$ into \cref{eq:alphanorm} to get \cref{eq:alphanormstar}, almost everywhere.
\end{proof}

\subsubsection{MCMC for Hierarchical Bayesian Models}

A hierarchical Bayesian model of some problem, including inverse problems, leads to the posterior distribution for unknown parameters $\theta$ conditioned on measured data $\mathbf{d}$, given by Bayes' rule 
\begin{equation}
\pi(\theta|{\bf d}) = \frac{\pi({\bf d}|\theta)\pi_\mathrm{p}(\theta)}{\pi({\bf d})} 
.
\label{eq:post}
\end{equation}

In the language of Bayesian analysis, $\pi_\mathrm{p}(\theta)$ is the \textit{prior} distribution, $\pi({\bf d}|\theta)$ as a function of $\theta$ is the \textit{likelihood} function, and $\pi({\bf d})$ is a normalising constant commonly referred to as the \textit{evidence}. The likelihood function is induced by the data-generating model
\begin{equation}
    {\bf d} = \mathcal F(\theta) + \epsilon
\end{equation}
where $\mathcal F(\theta)$ is the forward model and $\epsilon$ is the measurement error. When the measurement error is Gaussian, i.e., $\epsilon \sim \mathcal N(0, \Sigma_\epsilon)$, the particular likelihood function is proportional to
\begin{equation} \label{eq:likelihood} 
\mathcal{L}({\mathbf d} | \theta) = \exp \left( -\frac{1}{2} (\mathcal{F}(\theta) - {\mathbf d})^T {\mathbf \Sigma}_\epsilon^{-1} (\mathcal{F}(\theta) - {\mathbf d}) \right).
\end{equation}

In the Bayesian framework, solving the inverse problem is performed by exploring the posterior distribution $\pi(\theta|{\bf d})$ defined by \cref{eq:post} and evaluating statistics with respect to that distribution. Sample-based inference does this by drawing samples from the posterior distribution to evaluate sample-based Monte Carlo estimates of expected values. 
The plain vanilla route to drawing samples from $\pi(\theta|{\bf d})$ is to invoke \textbf{MH} Alg.~1 with $\pi_\text{t}(\cdot)=\pi(\cdot|{\bf d})$  such that
\[ 
\big[\theta^{1},\ldots, \theta^{N}\big] =\mathbf{MH}\left(\pi(\theta|{\bf d}),q(\cdot|\cdot),\theta^{0},N \right) . 
\]
Asymptotically, the density of the $j$th state $\theta^{j}$ converges to the posterior density $\pi(\cdot|{\bf d})$ and averages over this chain converge to expectations with respect to $\pi(\cdot|{\bf d})$, asymptotically in $N$.

\begin{remark}
When $\pi({\bf d})$ in \cref{eq:post} is finite, the Metropolis ratio ${\pi_\text{t}(\psi)}/{\pi_\text{t}(\theta^{j})}$ in Alg.~1 Eq.~\cref{eq:alphanorm} may be evaluated as a ratio of unnormalized densities
\begin{equation}\label{eq:mhunnorm}
      \frac{\pi({\bf d}|\psi)\pi_p(\psi)}{\pi({\bf d}|\theta^{j})\pi_p(\theta^{j})}. 
\end{equation}
Substitute  $\pi_\text{t}(\cdot)=\pi(\cdot|{\bf d})$ from Eq.~\cref{eq:post} into the Metropolis ratio and note that the normalisation constants $1/\pi({\bf d})$ in the numerator and in the denominator cancel.
Hereafter, for brevity we typically write the acceptance probability using the ratio of normalized posterior densities, as in Eq.~\cref{eq:alphanorm}, but actually compute with unnormalized densities, as in Eq.~\cref{eq:mhunnorm}.
\end{remark}

\subsubsection{Delayed Acceptance MCMC}
The Delayed Acceptance (DA) algorithm was introduced by Christen and Fox in \cite{Chr05}, with the goal of reducing the computational cost per iteration by utilizing a computationally cheaper approximation of the forward map, and thus also of the posterior density, for evaluating the acceptance probability in Alg.~1. One may also view DA as a way to improve the proposal kernel $q$, since DA modifies the proposal kernel using a Metropolis--Hastings accept-reject step to give an effective  proposal that is in detailed balance with an (approximate) distribution that is hopefully closer to the target than is the equilibrium distribution of the original proposal kernel.

The delayed acceptance algorithm is given in Alg.~2, for target (fine) density $\pi_\mathrm{F}$ and approximate (coarse) density $\pi_\mathrm{C}$.
Delayed acceptance  first performs a standard Metropolis--Hastings accept/reject step (as given in Alg.~1) with the approximate/coarse density $\pi_\mathrm{C}$. If accepted, a second accept reject/step is used, with acceptance probability chosen such that the composite iteration satisfies detailed balance with respect to the desired target $\pi_\mathrm{F}$.

In Alg.~2 Eq.~\cref{eq:beta_da}, $q_\mathrm{C}(\cdot|\cdot)$ is the effective proposal density from the first Metropolis--Hastings step with coarse density $\pi_\mathrm{C}(\cdot)$ as target; see~\cite{Chr05} for details. The acceptance probability in Eq.~\cref{eq:beta_da} is the standard Metropolis--Hastings rule for proposal density $q_\mathrm{C}$, targeting $\pi_{\mathrm{F}}(\cdot)$, hence Alg.~2 simulates a kernel in detailed balance with $\pi_{\mathrm{F}}(\cdot)$ and produces a chain that is ergodic with respect to 
$\pi_{\mathrm{F}}(\cdot)$; see~\cite{Chr05} for conditions on the approximation that ensure that the ergodic theorem applies. Computational cost per iteration is reduced because for proposals that are rejected in the first \textbf{MH} step in Eq.~\cref{eq:damh}, and thus result in $\psi=\theta^{j}$, the second acceptance ratio in Eq.~\cref{eq:beta_da} involving the more expensive, fine target density $\pi_{\mathrm{F}}(\cdot)$ does not need to be evaluated again.
\begin{center}\vspace{0.5cm}
\fbox{%
  \parbox{0.975\textwidth}{
    \textbf{Algorithm 2. Delayed Acceptance (DA)}  \\[1ex]
    \textbf{function}: $ 
                               \left[\theta^{1},\ldots, \theta^{N}\right] 
                        =\mathbf{DA}\left(\pi_\mathrm{F}(\cdot),\pi_\mathrm{C}(\cdot), q(\cdot|\cdot),\theta^{0},N \right)  $\\[1ex]
\begin{tabular}{@{}l@{}p{0.875\textwidth}}
    \textbf{input: }  & target (fine) density $\pi_\mathrm{F}(\cdot)$, approximate (coarse) density $\pi_\mathrm{C}(\cdot)$, proposal kernel $q(\cdot|\cdot)$, initial state $\theta^{0}$, number of steps $N$\\
    \textbf{output: } & ordered list of states $\left[\theta^{1},\ldots, \theta^{N}\right]$ \big(or just the final state $\theta^{N}$\big)
    \end{tabular} \\[1ex]
\textbf{for} $j = 0$ to $N-1$ : 
	\begin{itemize}
		\item Given $\theta^{j}$, generate proposal $\psi$ by invoking one step of \textbf{MH} Alg.~1 for coarse target $\pi_\mathrm{C}$:
		\begin{equation}
		  \psi = \mathbf{MH}\left(\pi_\mathrm{C}(\cdot),q(\cdot|\cdot),\theta^{j},1 \right) . 
		  \label{eq:damh}
		\end{equation}
        \item Accept proposal $\psi$ as the next state, i.e. set $\theta^{j+1} = \psi$, with probability 
		\begin{equation}
                 \alpha(\psi|\theta^{j}) = \min \left\{1, \frac
                 {\pi_{\mathrm{F}}(\psi)q_\mathrm{C}(\theta^{j}|\psi)}
                 {\pi_{\mathrm{F}}(\theta^{j})q_\mathrm{C}(\psi|\theta^{j})} 
                 \right\}
                 \label{eq:beta_da}
         \end{equation}
         otherwise reject proposal $\psi$ and set $\theta^{j+1} = \theta^{j}$.  
	\end{itemize}
  }
}\vspace{0.5cm}
\end{center}

In the multilevel context with levels indexed by $\ell$, the original \textbf{DA} Alg.~2 is a two-level method. Denote the more accurate forward map that defines the fine posterior distribution $\pi_{\ell}(\theta_\ell| {\bf d}_\ell)$ by $\mathcal F_\ell$, and the less accurate forward map that defines the approximate (coarse) posterior distribution $\pi_{\ell-1}(\theta_\ell|{\bf d}_{\ell-1})$ by $\mathcal F_{\ell-1}$. Note that we also allow a possibly altered or reduced data set ${\bf d}_{\ell-1}$ on level $\ell-1$, but that the states in the two forward maps and in the two distributions are the same. Then setting $\pi_\mathrm{F}(\cdot)=\pi_{\ell}(\cdot| {\bf d}_\ell)$ and $\pi_\mathrm{C}(\cdot)=\pi_{\ell-1}(\cdot| {\bf d}_{\ell-1})$ in the call to \textbf{DA}  Alg.~2, such that
\[ 
\big[\theta_{\ell}^{1},\ldots, \theta_{\ell}^{N}\big] =\mathbf{DA}\left(\pi_{\ell}(\cdot| {\bf d}_\ell),\pi_{\ell-1}(\cdot| {\bf d}_{\ell-1}),q(\cdot|\cdot),\theta^{0},N \right), 
\]
computes a chain that is ergodic with respect to $\pi_{\ell}(\cdot| {\bf d}_\ell)$, asymptotically as $N\to \infty$. 

\textbf{DA} Alg.~2 actually allows for the  approximate, coarse posterior distribution to depend on the state of the chain. Denote the state-dependent, approximate forward map at state $\theta$ by $\mathcal F_{\ell-1,\theta}$ and the resulting approximate posterior density by $\pi_{\ell-1,\theta}(\cdot|{\bf d}_{\ell-1})$. For state-dependent approximations it is always desirable and easy to achieve (see \cite{cui_posteriori_2019}) that $\mathcal F_{\ell-1,\theta}(\theta)=\mathcal F_{\ell}(\theta)$, so that $\pi_{\ell-1,\theta}(\theta|{\bf d}_{\ell-1})=k \pi_{\ell}(\theta|{\bf d}_{\ell})$ with the normalising constant $k$ independent of state $\theta$. The acceptance probability Eq.~\cref{eq:beta_da} then has the explicit form
\begin{equation}
                 \alpha(\psi|\theta^{j}) = \min \left\{1, \frac
                 {\min\left\{ \pi_\mathrm{F}(\psi)q(\theta^{j}|\psi) , 
                 \pi_{\mathrm{C},\psi}(\theta^{j})q(\psi|\theta^{j})
                 \right\} } 
                 {\min\left\{ \pi_\mathrm{F}(\theta^{j})q(\psi|\theta^{j}) ,
                 \pi_{\mathrm{C},\theta_{\ell}^{j}}(\psi)q(\theta^{j}|\psi)
                 \right\} } 
                 \right\}.
                 \label{eq:beta_da_state-dep}
\end{equation}
For technical reasons, as explained in \cref{rem:state-dependent} below, we will not use state-dependent approximations, but rather restrict ourselves to fixed approximate forward maps that do not depend on the current state.

\subsection{Detailed Balance Beyond Two Levels} \label{sec:mlda_convergence}

We will now extend \textbf{DA} to randomised-length-subchains, to embedded state spaces at the coarser level, and finally to multiple levels. The resulting Markov chain on the finest level is shown to be in detalied balance with the target density.

\subsubsection{Randomised-Length-Subchain Surrogate Transition MCMC}\label{sec:rst}
 
When the approximate forward map does not depend on the current state -- for example, when using a fixed coarse discretization for a PDE -- the resulting approximate posterior density is a fixed \emph{surrogate} for the true posterior density, and Alg.~2 coincides with the surrogate transition method introduced by Liu \cite{Liu08}.  \Cref{lem:cf2} then implies that the acceptance probability in Eq.~\cref{eq:beta_da} is
\begin{equation}
                 \alpha(\psi|\theta^{j}) = 
                 \min \left\{1, \frac{\pi_{\mathrm{F}}(\psi)\pi_{\mathrm{C}}(\theta^{j}) } 
                 {\pi_{\mathrm{F}}(\theta^{j})\pi_{\mathrm{C}}(\psi) } 
                 \right\}\, ,
                 \label{eq:beta_st}
\end{equation}
since the Metropolis--Hastings step in Eq.~\cref{eq:damh} ensures that the effective proposal kernel $q_\mathrm{C}(\cdot|\cdot)$ is in detailed balance with the approximate density $\pi_{\mathrm{C}}(\cdot)$. 

We extend the surrogate transition method in two ways. As noted by Liu~\cite{Liu08}, multiple steps can be made with the surrogate, i.e. iterating the proposal and first accept/reject step Eq.~\cref{eq:damh} before performing the second accept/reject step with acceptance probability in Eq.~\cref{eq:beta_st}. We call the sequence of states generated by multiple steps of Eq.~\cref{eq:damh} a \emph{subchain}. Further, we consider subchains of random length, set according to a probability mass function (pmf) $p(\cdot)$ on the positive integers. In practice we set $J\in\mathbb{Z}^+$ and then set $p = \mathcal U(\{1,2,\ldots,J\})$, though note that a deterministic choice of subchain length is another special case. The utility of randomising the subchain length will become apparent in \cref{sec:VarianceReduction}. These extensions are included in Alg.~3.
\begin{center}\vspace{0.5cm}
\fbox{%
  \parbox{0.975\textwidth}{
    \textbf{Algorithm 3. Randomised-Length-Subchain Surrogate Transition (RST)}\\[1ex]  
        \textbf{function}: $ 
                               \left[\theta^{1},\ldots, \theta^{N}\right] 
                        =\mathbf{RST}\left(\pi_\mathrm{F}(\cdot),\pi_\mathrm{C}(\cdot), q(\cdot|\cdot),p(\cdot),\theta^{0},N \right) $\\[1ex]
\begin{tabular}{@{}l@{}p{0.88\textwidth}}
    \textbf{input: }  & target (fine) density $\pi_\mathrm{F}(\cdot)$, surrogate (coarse) density $\pi_\mathrm{C}(\cdot)$, proposal kernel $q(\cdot|\cdot)$, probability mass function $p(\cdot)$ over subchain length, initial state $\theta^{0}$, number of steps $N$\\
    \textbf{output: } & ordered list of states $\left[\theta^{1},\ldots, \theta^{N}\right]$ \big(or just the final state $\theta^{N}$\big)
    \end{tabular} \\[1ex]
\textbf{for} $j = 0$ to $N-1$ :         
	\begin{itemize}
	\item Draw the subchain length $n \sim p(\cdot)$.  
	\item Starting at $\theta^{j}$, generate subchain of length $n$ using \textbf{MH}  Alg.~1 to target $\pi_\mathrm{C}(\cdot)$:
	\begin{equation}
	   \psi = \mathbf{MH}\left(\pi_\mathrm{C}(\cdot),q(\cdot|\cdot),\theta^{j},n \right)  \label{eq:stmh}
	\end{equation}
	\item Accept the proposal $\psi$ as the next sample, i.e. set $\theta^{j+1} = \psi$, with probability
	\begin{equation}
	\alpha(\psi|\theta^{j}) = 
                 \min \left\{1, \frac{\pi_{\mathrm{F}}(\psi)\pi_{\mathrm{C}}(\theta^{j}) } 
                 {\pi_{\mathrm{F}}(\theta^{j})\pi_{\mathrm{C}}(\psi) } 
                 \right\}.
	\label{eq:beta_rst}
	\end{equation}
	 otherwise reject and set $\theta^{j+1} = \theta^{j}$.
	\end{itemize}
	}
	}
\vspace{0.5cm}	
\end{center}

We will show that Alg.~3 satisfies detailed balance using \cref{lem:commute}, needed also later.
\begin{definition} 
 We define composition of Markov kernels $K_1$ and $K_2$ in the usual way \cite{giry1982categorical} by
 \[ ( K_1\circ K_2)(\theta|\psi) = \int K_1(\theta|\phi) K_2(\phi|\psi) \dd \phi. \]
\end{definition}
Composition is associative, by Tonelli's theorem, so, by induction, the composition of multiple Markov kernels is well defined. The composition of a kernel $K$ with itself will be denoted $K^2$, while the composition of $n$ lots of the kernel $K$ is denoted $K^n$, so the notation is the same as for composition of transition matrices defining Markov processes with a finite state space.
\begin{lemma} \label{lem:commute}
  Let $K_1(x|y)$ and $K_2(x|y)$ be two transition kernels that are in detailed balance with a density $\pi$ and that commute. Then their composition 
 $(K_1\circ K_2)$ is also in detailed balance with $\pi$.
\end{lemma}
\begin{proof}
\begin{align*}
 & \\[-6.5ex]
  \qquad\quad\pi(\psi)( K_1\circ K_2)(\theta|\psi) &= \pi(\psi) \int K_1(\theta|\phi) K_2(\phi|\psi) \dd \phi                    = \pi(\psi) \int K_2(\theta|\phi) K_1(\phi|\psi)  \dd \phi\\
                             &= \pi(\psi) \int K_2(\phi|\theta)\frac{\pi(\theta)}{\pi(\phi)} K_1(\psi|\phi)\frac{\pi(\phi)}{\pi(\psi)}   \dd \phi\\
                             &= \pi(\theta) \int K_2(\phi|\theta) K_1(\psi|\phi)   \dd \phi
                            = \pi(\theta) (K_1\circ K_2)(\psi|\theta)\\[-4.5ex]
\end{align*}
\end{proof}

\begin{lemma} \label{lem:rst}
 Alg.~3 simulates a Markov chain that is in detailed balance with $\pi_{\mathrm{F}}(\cdot)$.
\end{lemma}
\begin{proof}
  Recall that the effective density $q_{\mathrm{C}}(\cdot|\cdot)$ for proposals drawn according to Alg.~2 Eq.~\cref{eq:damh} is in detailed balance with $\pi_{\mathrm{C}}(\cdot)$. Since $q_\mathrm{C}$ clearly commutes with itself, using \cref{lem:commute}, it follows by induction that $q_{\mathrm{C}}^n(\cdot|\cdot)$, (i.e. $q_{\mathrm{C}}$ composed $n$ times with itself) is in detailed balance with $\pi_{\mathrm{C}}(\cdot)$ for any $n$. Hence, the effective proposal density induced by Alg.~3 Eq.~\cref{eq:stmh}, namely the mixture kernel $\sum_{n\in\mathbb{Z}^+} p(n) q_{\mathrm{C}}^n(\cdot|\cdot)$
  is also in detailed balance with $\pi_{\mathrm{C}}(\cdot)$.
  
  Finally, the acceptance probability in Alg.~3 Eq.~\cref{eq:beta_rst} for target density $\pi_\mathrm{F}(\cdot)$ follows from \cref{lem:cf2}, since the proposal kernel is in detailed balance with $\pi_\mathrm{C}(\cdot)$. Consequently, Alg.~3 produces a chain in detailed balance with $\pi_\mathrm{F}(\cdot)$.
\end{proof}

\begin{remark}
Choosing a multinomial pmf over the subchain length, with $p(J) = 1$ and $p(\lnot J) = 0$, implies that \cref{lem:rst} is also valid for the special case of a fixed subchain length $J_C$.
\end{remark}

\begin{remark}\label{rem:state-dependent}
We do not yet have a version of \cref{lem:rst} for fully state-dependent approximations, which is why we restrict here to state-independent surrogates.
\end{remark}

\begin{remark}
 If the densities of the coarse and fine posterior distributions in Alg.~3 are with respect to the same prior distribution, i.e. $\pi_\mathrm{F}(\theta)= \pi_{\ell}(\theta|{\bf d}_{\ell})\propto \pi_{\ell}({\bf d}_{\ell}|\theta)\pi_\mathrm{p}(\theta)$ and  $\pi_\mathrm{C}(\theta)= \pi_{\ell-1}(\theta|{\bf d}_{\ell-1})\propto \pi_{\ell-1}({\bf d}_{\ell-1}|\theta)\pi_\mathrm{p}(\theta)$, the acceptance probability in Alg.~3 Eq.~\cref{eq:beta_rst} is equal to
 	\begin{equation}
	\alpha(\psi|\theta^{j}) = 
                 \min \left\{1, \frac{\pi_{\ell}\big(\mathbf{d}_{\ell}|\psi\big)\pi_{\ell-1}\big(\mathbf{d}_{\ell-1}|\theta^{j}\big) } 
                 {\pi_{\ell}\big(\mathbf{d}_{\ell}|\theta^{j}\big)\pi_{\ell-1}\big(\mathbf{d}_{\ell-1}|\psi\big) } 
                 \right\}.
	\label{eq:beta_rstv1}
	\end{equation}
\end{remark}

\subsubsection{Different Fine and Coarse States}\label{sec:embedding}

In delayed acceptance Alg.~2, and hence also in the randomised surrogate transition Alg.~3, the state in the fine and coarse target distributions is the same. In the MLMCMC of Dodwell \textit{et al.}  \cite{dodwell_hierarchical_2015} different levels can have different states, which is natural when using e.g. a hierarchy of FEM discretisations with different levels of mesh refinement. In this context, the states at different levels form a hierarchy of embedded spaces, where the state vector at any given level is part of the state vector at the next finer level. Hence, in a two-level hierarchy as described above, the (fine) state $\theta$ can be partitioned into ``coarse modes'' (or ``components'') denoted $\theta_{\mathrm{C}}$ and ``fine modes'' $\theta_{\mathrm{F}}$, so that $\theta=\left( \theta_{\mathrm{F}},\theta_{\mathrm{C}}\right) $. The coarse modes $\theta_\mathrm{C}$ are the components of the state vector on the coarse, approximate level targeted by $\pi_\mathrm{C}$, while the fine target distribution $\pi_\mathrm{F}$ also depends on the fine modes $\theta_\mathrm{F}$.

The randomised surrogate transition Alg.~3 is easily extended to allow this structure, as shown in Alg.~4 below, where surrogate transition is only used to propose the states of the coarse modes,  while the fine modes are drawn from some additional proposal distribution. The composite of the fine and coarse proposals then forms the proposed state at the fine level. For this extension it is important that the fine modes are proposed independently of the coarse modes to ensure detailed balance, as shown below.

\begin{lemma} \label{lem:tlda}
Two Level Delayed Acceptance in Alg.~4 generates a chain in detailed balance with $\pi_\mathrm{F}$.
\end{lemma}
\begin{proof}
As noted in the proof of \cref{lem:rst}, the proposal density $q_\mathrm{C}$ induced by the surrogate transition step in Alg.~4 Eq.~\cref{eq:tldamh} is in detailed balance with the coarse target density $\pi_\mathrm{C}(\cdot)$ over $\theta_\mathrm{C}$. As a kernel on the composite state $\theta=(\theta_\mathrm{F},\theta_\mathrm{C})$ we can write the coarse proposal as
\[ 
K_\mathrm{C} = \left[
    \begin{array}{c;{2pt/2pt}c}
        I & 0 \\ \hdashline[2pt/2pt]
        0 & q_\mathrm{C} 
    \end{array}
\right]  
\]
where $I$ denotes the identity of appropriate dimension.
Similarly, the fine proposal Eq.~\cref{eq:tldafp} on the composite state has kernel
\[ 
K_\mathrm{F} = \left[
    \begin{array}{c;{2pt/2pt}c}
        q_\mathrm{F} & 0 \\ \hdashline[2pt/2pt]
        0 &  I
    \end{array}
\right] . 
\]
Since $K_\mathrm{F}$ does not change the coarse modes, it trivially 
is in detailed balance with $\pi_\mathrm{C}(\cdot)$. Further, it is easy to check that $K_\mathrm{C}$ and $K_\mathrm{F}$ commute. Hence, by \cref{lem:commute} the composition  $(K_\mathrm{F}\circ K_\mathrm{C}^n)$ is also in detailed balance with  $\pi_\mathrm{C}(\cdot)$ and so is the effective proposal kernel $\sum_{n\in\mathbb{Z}^+}p(n)(K_\mathrm{F}\circ K_\mathrm{C}^n)$ for drawing $\psi=\left(\psi_\mathrm{F},\psi_\mathrm{C} \right)$ according to Alg.~4 Eqs.~\cref{eq:tldamh} and \cref{eq:tldafp}. The acceptance probability in Alg.~4 Eq.~\cref{eq:beta_tlda} then follows again from \cref{lem:cf2} and the chain produced by Alg.~4 is in detailed balance with $\pi_\mathrm{F}(\cdot)$, as desired.
\end{proof}
Note that the Randomised Surrogate Transition Alg.~3 is a special case of Alg.~4 with  $\theta^{j}=\theta_\mathrm{C}^{j}$, i.e. $\theta_\mathrm{F}^{j}$ is empty, and correspondingly $q_{\mathrm{F}}(\cdot|\cdot)$ is the (trivial) proposal on the empty space. 

\begin{center}\vspace{0.5cm}
\fbox{%
  \parbox{0.975\textwidth}{
    \textbf{Algorithm 4. Two Level Delayed Acceptance (TLDA)}\\[1ex]  
        \textbf{function}: $ 
                               \left[\theta^{1},\ldots, \theta^{N}\right] 
                        = \mathbf{TLDA}\left(\pi_\mathrm{F}(\cdot),\pi_\mathrm{C}(\cdot), q(\cdot|\cdot),q_{\mathrm{F}}(\cdot|\cdot),p(\cdot), \theta^{0},N \right)$\\[1ex]
\begin{tabular}{@{}l@{}p{0.88\textwidth}}
    \textbf{input: }  & target (fine) density $\pi_\mathrm{F}(\cdot)$, surrogate (coarse) density $\pi_\mathrm{C}(\cdot)$, proposal kernel $q(\cdot|\cdot)$ on coarse modes, proposal kernel  $q_{\mathrm{F}}(\cdot|\cdot)$ on fine modes, probability mass function $p(\cdot)$ over subchain length, initial state $\theta^{0}$, number of steps $N$\\
    \textbf{output: } & ordered list of states $\left[\theta^{1},\ldots, \theta^{N}\right]$ \big(or just the final state $\theta^{N}$\big)
    \end{tabular} \\[1ex]
\textbf{for} $j = 0$ to $N-1$ :        
	\begin{itemize}
	
	\item Draw the subchain length $n \sim p(\cdot)$.  
	\item Starting at $\theta_\mathrm{C}^{j}$, generate subchain of length $n$ using \textbf{MH} Alg.~1 to target $\pi_\mathrm{C}(\cdot)$:
	\begin{equation}
	  \psi_\mathrm{C} = \mathbf{MH}\left(\pi_\mathrm{C}(\cdot),q(\cdot|\cdot),\theta_\mathrm{C}^{j},n \right) 
	  \label{eq:tldamh}
	\end{equation}
	\item Draw the fine-mode proposal 
	\begin{equation}
	   \psi_\mathrm{F} \sim q_{\mathrm{F}}(\cdot|\theta_\mathrm{F}^{j}) 
	   \label{eq:tldafp}
	\end{equation}
	\item Accept proposal $\psi=(\psi_\mathrm{F},\psi_\mathrm{C})$ as next sample, i.e., set $\theta^{j+1} = \psi$, with probability
	\begin{equation}
	\alpha(\psi|\theta^{j}) = 
                 \min \left\{1, \frac{\pi_{\mathrm{F}}(\psi)\pi_{\mathrm{C}}(\theta_\mathrm{C}^{j}) } 
                 {\pi_{\mathrm{F}}(\theta^{j})\pi_{\mathrm{C}}(\psi_\mathrm{C}) } 
                 \right\}.
	\label{eq:beta_tlda}
	\end{equation}
	 otherwise reject and set $\theta^{j+1} = \theta^{j}$.
	\end{itemize}
	}
	}
\vspace{0.5cm}
\end{center}

\subsubsection{Multilevel Delayed Acceptance}\label{sec:MLDA}

The multilevel delayed acceptance algorithm is a recursive version of \textbf{TLDA} in which instead of invoking Metropolis--Hastings to generate a subchain at the coarser levels the algorithm is recursively invoked again (except for the coarsest level $\ell=0$), leading to a \textit{hierarchical} multilevel delayed acceptance algorithm, which admits an arbitrary number of model levels $L$. The flexibility with respect to the depth of the model hierarchy and the subchain lengths allows for tailoring the algorithm to various objectives, including the reduction of variance (see \cref{sec:VarianceReduction}) or increasing the effective sample size (see \cref{sec:gravity}).

To be more precise, \textbf{MLDA} Alg.~5 below is called on the most accurate, finest level $L$. Then, for levels $1\leq\ell\leq L$ it generates a subchain at level $\ell-1$ as in \textbf{TLDA}, by recursively invoking \textbf{MLDA} on level $\ell-1$, until the coarsest level $\ell=0$ is reached where plain \textbf{MH} in invoked.
Required for \textbf{MLDA} are the hierarchy of density functions $\pi_{0}(\cdot),\ldots,\pi_{L}(\cdot)$ along with a coarsest-level proposal $q_0$, partitions into coarse and fine modes at each level, fine-mode proposals $q_{1,\mathrm{F}},\ldots,q_{L,\mathrm{F}}$ and probability mass functions $p_{1}(\cdot),\ldots,p_{L}(\cdot)$ over the subchain lengths on levels 0 to $L-1$. Note that the fine-mode proposals are used to draw the \textit{additional} finer modes on each level $1 \leq \ell \leq L$, to construct a hierarchy of embedded spaces as explained in \cref{sec:embedding}. The algorithm is illustrated conceptually in \cref{fig:mlda-concept}.

\begin{center}\vspace{0.5cm}
\fbox{%
  \parbox{0.975\textwidth}{
    \textbf{Algorithm 5. Multilevel Delayed Acceptance (MLDA):}\\[1ex]
            \textbf{function}: $ 
                               \left[\theta_\ell^{1},\ldots, \theta_\ell^{N}\right]
                        =\mathbf{MLDA}\left( \left\{\pi_{k}\right\}_{k=0}^{\ell}, q_0,\left\{q_{k,\mathrm{F}}\right\}_{k=1}^{\ell}, \left\{p_{k}\right\}_{k=1}^{\ell},\theta_\ell^{0},\ell,N \right)  $\\[1ex]
\begin{tabular}{@{}l@{}p{0.88\textwidth}}
    \textbf{input: }  & target densities $\pi_{0}(\cdot),\ldots\pi_{\ell}(\cdot)$, proposal densities $q_0(\cdot|\cdot)$ and  $q_{1,\mathrm{F}}(\cdot|\cdot),\ldots,q_{\ell,\mathrm{F}}$, probability mass functions $p_{1}(\cdot),\ldots,p_{\ell}(\cdot)$ over subchain lengths on levels $0$ to $\ell-1$, initial state $\theta_\ell^{0}$, current level index $\ell$, number of steps $N$\\
    \textbf{output: } & ordered list of states $[\theta_\ell^{1},\ldots, \theta_\ell^{N}]$ at level $\ell$ \big(or just the final state $\theta_\ell^{N}$\big)
    \end{tabular} \\[1ex]
\textbf{for} $j = 0$ to $N-1$ : 
	\begin{itemize}

	\item Draw the subchain length $n_{\ell} \sim p_{\ell}(\cdot)$ for level $\ell-1$.  
	\item Starting at $\theta_{\ell,\mathrm{C}}^{j}$, generate a subchain of length $n_{\ell}$ on level $\ell-1$:
	\begin{itemize}
	 \item If $\ell=1$, use the Metropolis--Hastings algorithm to generate the subchain
	 	\[ \psi_\mathrm{C} = \mathbf{MH}\left(\pi_0(\cdot),q_0(\cdot,\cdot),\theta_{1,\mathrm{C}}^{j},n_1 \right) . \]
    \item If $\ell>1$, generate the subchain by (recursively) calling \textbf{MLDA}
        \[ \psi_\mathrm{C} = \mathbf{MLDA}\left( \left\{\pi_{k}(\cdot)\right\}_{k=0}^{\ell-1},q_0(\cdot|\cdot),\left\{q_{k,\mathrm{F}}\right\}_{k=1}^{\ell-1},\left\{p_{k}\right\}_{k=1}^{\ell-1},\theta_{\ell,\mathrm{C}}^{j},\ell-1,n_{\ell} \right) . \]
	\end{itemize}
	\item Draw the fine-mode proposal 
	\ $\psi_\mathrm{F} \sim q_{\ell,\mathrm{F}}\big(\cdot|\theta_{\ell,\mathrm{F}}^{j}\big).$
	\item Accept proposal $\psi=(\psi_\mathrm{F},\psi_\mathrm{C})$ as next sample, i.e., set $\theta_\ell^{j+1} = \psi$, with probability
	\begin{equation}
	\alpha(\psi|\theta^{j}) = 
                 \min \left\{1, \frac{\pi_{\ell}(\psi)\pi_{\ell-1}\big(\theta_{\ell,\mathrm{C}}^{j}\big) } 
                 {\pi_{\ell}\big(\theta_{\ell}^{j}\big)\pi_{\ell-1}(\psi_\mathrm{C}) } 
                 \right\}
	\label{eq:beta_mlda}
	\end{equation}
	 otherwise reject and set $\theta_\ell^{j+1} = \theta_\ell^{j}$.
	 
	\end{itemize}
	}
	}
\vspace{0.5cm}	
\end{center}

\begin{figure}[t]
    \centering
    \includegraphics[width=0.8\linewidth]{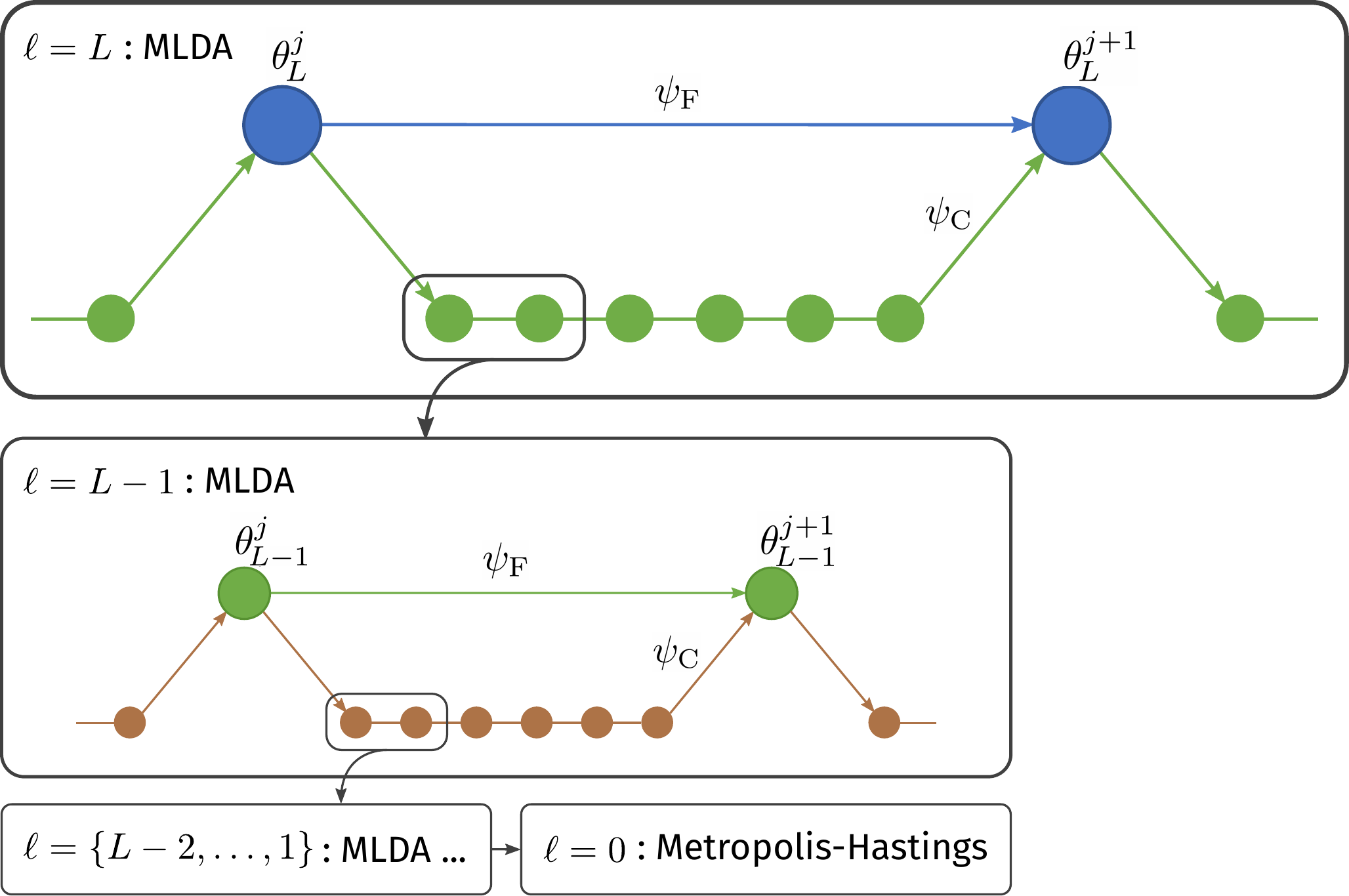}
    \caption{The MLDA algorithm sampling with a model hierarchy with $L$ levels. The MLDA sampler is employed recursively on each level $\ell > 0$, while on level $\ell=0$, any Metropolis--Hastings algorithm can be used. On level $\ell=L$, the MLDA sampler generates a Markov chain in detailed balance with $\pi_L$, according to \cref{thm:mlda}. On each level $\ell < L$, the respective samplers generate proposals for the coarse modes $\psi_\mathrm{C}$ of the next-finer level.}
    \label{fig:mlda-concept}
\end{figure}

A chain of length $N$ at level $L$ is then produced by calling
\begin{equation}
   \left[\theta_L^{1},\ldots, \theta_L^{N}\right]
                        =\mathbf{MLDA}\left( \left\{\pi_{k}\right\}_{k=0}^{L}, q_0,\left\{q_{k,\mathrm{F}}\right\}_{k=1}^{L}, \left\{p_{k}\right\}_{k=1}^{L},\theta_L^{0},L,N \right).
  \label{eq:mldacall}
\end{equation}
We can now state the main theoretical result of paper. 
\begin{theorem}\label{thm:mlda}
Multilevel Delayed Acceptance in Alg. 5, invoked as in~\cref{eq:mldacall}, generates a Markov chain that is in detailed balance with $\pi_L$.
\end{theorem}
\begin{proof}
The proof follows essentially by induction on the level $\ell$ from the proof of \cref{lem:tlda}. At level $\ell=1$, \textbf{MLDA} is equivalent to \textbf{TLDA}, and so the base step follows immediately from \cref{lem:tlda}. Let us now assume that the proposal kernel for $\psi=(\psi_\mathrm{F},\psi_\mathrm{C})$ on level $\ell$ simulated using \textbf{MLDA} on level $\ell-1$ is in detailed balance with $\pi_{\ell-1}$. Then it follows from \cref{lem:cf2} that the acceptance probability in Alg.~5 Eq.~\cref{eq:beta_mlda} produces a Markov chain that is in detailed balance with $\pi_\ell(\cdot)$, which concludes the induction step.
\end{proof}

\subsubsection{Comparison with MLMCMC} \label{sec:mlmcmc_comparison}

The generalisation of Delayed Acceptance to an extended multilevel setting leads to clear similarities with the Multilevel Markov Chain Monte Carlo (MLMCMC) Method proposed by Dodwell {\em et al.}~\cite{dodwell_hierarchical_2015}. The more subtle difference between the two approaches is illustrated in \cref{fig:mlda_vs_mlmcmc}. 

The MLDA algorithm can be seen as a recursive application of the surrogate transition method over multiple levels. If a proposal $\psi$ from level $\ell-1$ for level $\ell$ at state $\theta_\ell^{j}$ is rejected, the initial state for the coarse subchain $\theta_{\ell-1}^{0}$ is set back to $\theta_\ell^{j}$. Hence, the new coarse subchain, which will generate the next proposal for level $\ell$, is initialised from the same state as the previous subchain. 
\begin{figure}[t]
\includegraphics[width = 0.56\linewidth]{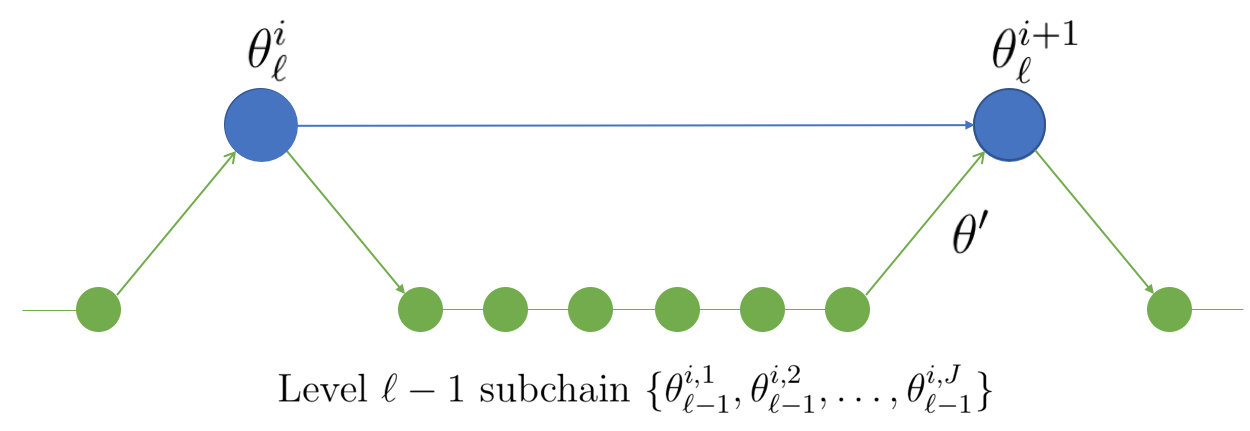}
\includegraphics[width = 0.4\linewidth]{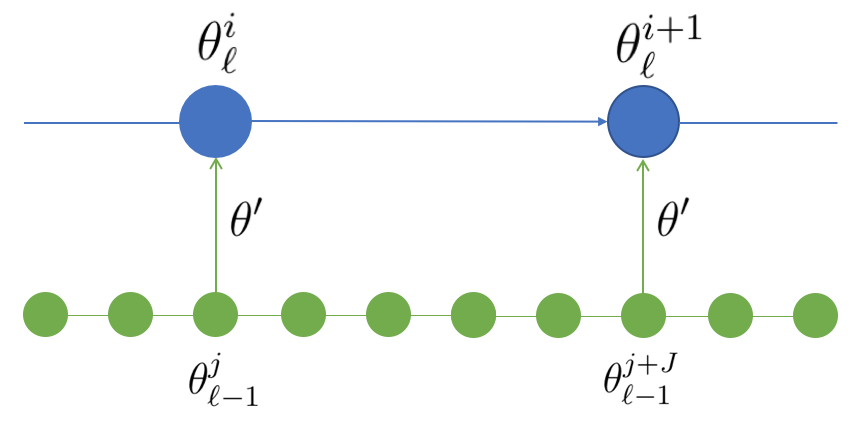}
\caption{Schematic for generating a proposal $\theta'$ on level $\ell$ for MLDA (left) and  MLMCMC (right) using a fixed length subchain of length $J$. The key difference is that for MLMCMC the coarse chain on level $\ell-1$ is generated independently of the chain on level $\ell$.}
\label{fig:mlda_vs_mlmcmc}
\end{figure}

For MLMCMC \cite{dodwell_hierarchical_2015}, even if the coarse proposal is rejected, the  coarse chain continues independently of the fine chain. In analogy to the subchain picture in MLDA, this corresponds to initialising the subchain on level $\ell-1$ 
with the coarse state $\psi_\mathrm{C}$ that has just been rejected on level $\ell$. As a result, coarse and fine chains will separate
and only re-coalesce once a coarse proposal is accepted at the fine level. This choice provides better mixing at coarse levels and allows for efficient parallelisation of the MLMCMC algorithm  \cite{10.1145/3458817.3476150}, but it does entail one important caveat; 
The practical algorithm in \cite[Alg.~3]{dodwell_hierarchical_2015}
does not necessarily define a Markov process unless coarse proposals passed to the next finer level are independent, as in \cite[Alg.~2]{dodwell_hierarchical_2015}. The practical implication of violating this requirement is that we do not have a proof of convergence of MLMCMC with finite subchains because we cannot apply the theorems that guarantee convergence for homogeneous Markov chains. Indeed, numerical experiments (not shown) indicate that estimates using MLMCMC with finite subchains are biased and that the underlying chains do not converge to the desired target distributions.

Accordingly, in theory the practical multilevel estimator proposed by Dodwell \textit{et al.} \cite[Alg.~3]{dodwell_hierarchical_2015} is only unbiased if the coarse proposal is an independent sample from $\pi_{\ell-1}$; therefore only at infinite computational cost (i.e. when the subchain length goes to infinity). However, if the fixed subchain length is chosen to be greater than twice the integrated autocorrelation length of the chain at that level, in practice this bias disappears. This imposes the constraint that the subchain length might have to be fairly long. If the acceptance rate is also relatively low, the method becomes computationally inefficient, i.e. a lot of computational effort has to be put into generating independent proposals from a coarse distribution which are then rejected with high probability.

\subsection{A Multilevel Estimator and Variance Reduction} \label{sec:VarianceReduction}

Using the MLDA sampler proposed above, it is in fact possible to define an asymptotically unbiased multilevel estimator that retains most of the computational benefits of both Multilevel Monte Carlo \cite{giles_multilevel_2008} and MLMCMC \cite{dodwell_hierarchical_2015}. Let $Q_\ell(\theta_\ell)$ define some quantity of interest computed on level $\ell = 0, \ldots,L$. The aim is to estimate $\mathbb E_{\pi_L}[Q_L]$ -- the expectation of $Q_L$ with respect to the posterior distribution $\pi_L$ on the finest level $L$ -- using as little computational effort as possible. 

The idea of Multilevel Monte Carlo is, at its heart, very simple. The key is to avoid estimating the expected value $\mathbb E_\ell[Q_\ell]$ directly on level $\ell$, but instead to estimate the correction with respect to the next lower level. Under the assumption that samples on level $\ell - 1$ are cheaper to compute than on level $\ell$ and that the variance of the correction term is smaller than the variance of $Q_\ell$ itself, the cost of computing this estimator is much lower than an estimator defined solely on samples from level $\ell$.
In the context of MLDA and MLMCMC, the target density $\pi_\ell$ depends on $\ell$, so that we write
\begin{equation}\label{eqn:ML_Estimator}
\mathbb E_{\pi_L}[Q_L] = \mathbb E_{\pi_0}[Q_0] + \sum_{\ell=1}^L\Big( \mathbb E_{\pi_\ell}[Q_\ell] - \mathbb E_{\pi_{\ell-1}}[Q_{\ell-1}]\Big), 
\end{equation}
which is achieved by adding and subtracting $\mathbb E_{\pi_\ell}[Q_\ell]$ for all levels $\ell = 0,\ldots, L-1$. Note that for the particular case where the densities $\{\pi_\ell\}_{\ell=0}^{L}$ are all equal, this reduces to the simple telescoping sum forming the basis of standard Multilevel Monte Carlo \cite{giles_multilevel_2008}. 

The practical MLMCMC algorithm in \cite[Alg.~3]{dodwell_hierarchical_2015} now proceeds by estimating 
the first term in \cref{eqn:ML_Estimator} using the MCMC estimator $E_{\pi_0}[Q_0] \approx \frac{1}{N_0}\sum_{i = 1}^{N_0} Q_0(\theta_0^{i})$ with a Markov chain $\big[\theta_0^{1},\ldots,\theta_0^{N_0}\big]$ produced with a standard \textbf{MH} on the coarsest level. Each of the correction terms for $\ell \geq 1$ is estimated by
\begin{equation}\label{eqn:ML_Telescoping}
\mathbb E_{\pi_\ell}[Q_\ell] - \mathbb E_{\pi_{\ell-1}}[Q_{\ell-1}] \approx \frac{1}{N_\ell}\sum_{i = 1}^{N_\ell} Q_\ell \big(\theta_\ell^{i}\big) - Q_{\ell - 1}\big(\theta_{\ell-1}^{J_{\ell}i}\big),
\end{equation}
where $N_\ell$ is the total number of samples on level $\ell$ after subtracting burn-in, $J_{\ell}$ is the subchain length on level $\ell-1$ and $\theta_{\ell-1}^{J_{\ell}i}$ is the state of the coarse chain used as the proposal for the $i$th state of the fine chain in the MLMCMC algorithm. 
As mentioned in \cref{sec:mlmcmc_comparison}, this multilevel estimator is only unbiased for MLMCMC as $J_\ell \to \infty$ or, in practice, for coarse subchains with $J_\ell$ greater than twice  
the integrated autocorrelation length. 

An unbiased multilevel estimator can be produced  using MLDA, without this constraint on the subchain lengths. However, since the levels of MLDA are strongly coupled and the coarse levels are consecutively realigned with the next-finer, this is non-trivial. We achieve it by employing a particular form of RST Alg.~3 in the MLDA Alg.~5. For all $\ell = 1,\ldots,L$, we set the probability mass function over the subchain length on level $\ell-1$ to the discrete uniform distribution $p_{\ell} = \mathcal U(\{1,2,\ldots, J_{\ell}\})$, where $J_{\ell}$ is the maximum subchain length. Hence, the $j$th proposal $\psi_\mathrm{C} = \psi_{\ell-1}^{j}$ for the coarse modes on level $\ell$ in this version of MLDA constitutes an independent, uniformly-at-random draw from a subchain of length $J_{\ell}$ on level $\ell-1$. Crucially, we let the coarse sampler continue sampling beyond the proposed state to produce subchains of fixed length $J_{\ell}$ for each state of the fine chain. Moreover, we also evaluate and store the quantity of interest at each state of each of those subchains on level $\ell-1$.

Thus, using MLDA in this way to compute a chain $[\theta_L^{1},\ldots,\theta_L^{N}]$ on the finest level $L$. In addition to the 
\[
N_L = N \ \ \text{samples} \quad Q_L\big(\theta_L^{1}\big), \ldots, Q_L\big(\theta_L^{N_L}\big) \quad \text{on level} \ L,
\]
we obtain also
\[
N_\ell = N \times \prod_{k=\ell}^{L-1} J_{k+1} \ \ \text{samples} \quad Q_\ell\big(\theta_\ell^{1}\big), \ldots, Q_\ell\big(\theta_\ell^{N_\ell}\big) \quad \text{on levels} \ \ \ell = 0,\ldots, L-1.
\]
Using those samples the following asymptotically unbiased MLDA estimator of the posterior expectation  $\mathbb E_{\pi_L}[Q_L]$ can be defined:
\begin{equation}
\label{eq:mlda}
\widehat{Q}_L := \frac{1}{N_0}\sum_{i = 1}^{N_0} Q_0\big(\theta_0^{i}\big) + \sum_{\ell=1}^L \frac{1}{N_\ell}\sum_{j = 1}^{N_\ell} Q_\ell\big(\theta_\ell^{j}\big) - Q_{\ell - 1}\big(\psi_{\ell-1}^{j}\big).
\end{equation}
Here, $\psi_{\ell-1}^{j}$ denotes the proposal $\psi_\mathrm{C}$ for the coarse modes of the $j$th state $\theta_\ell^{j}$ of the Markov chain on level $\ell$ produced by MLDA in Alg.~5.

Let us first discuss, why this estimator is asymptotically unbiased. For each $j$, the proposals $\psi_{l-1}^{j}$ are independently and uniformly drawn from the subchain $[\theta_{l-1}^{k} : (j-1)J_\ell < k \le jJ_\ell]$. Thus, the ensemble $\big\{\psi_{l-1}^{1},\ldots,\psi_{l-1}^{N_\ell}\big\}$ is a random draw from $\big\{\theta_{l-1}^{1},\ldots,\theta_{l-1}^{N_{\ell-1}}\big\}$ and thus identically distributed. As a consequence, in the limit as $N_\ell \to \infty$ for all $\ell$, most terms on the right hand side of \cref{eq:mlda} cancel. What remains, is $\sum_{j = 1}^{N_L} Q_L\big(\theta_L^{j}\big)$, which due to \cref{thm:mlda} is an unbiased estimator for $\mathbb E_{\pi_L}[Q_L]$ in the limit as $N_L \to \infty$.

Since the coarse subsamplers in MLDA are repeatedly realigned with the next finer distribution by way of the MLDA transition kernel, the samples on the coarse levels are in fact not distributed according to the ``vanilla'' densities $\{\pi_\ell\}_{\ell=0}^{L-1}$, but come from some ``hybrid'' mixture distributions.With the particular choice for $p_\ell$, the density of the mixture distribution arising from subsampling the coarse density on level $\ell-1 < L$ can be written
\begin{equation}
\tilde{\pi}_{\ell-1} = \frac{1}{J_\ell} \sum_{n=1}^{J_\ell} K_{\ell-1}^{n} \: \tilde{\pi}_{\ell,C}
\end{equation}
where $\tilde{\pi}_{\ell,C}$ is the marginal density of the coarse modes of the next finer density, $K_{\ell-1}$ is the transition kernel simulated by each step of subsampling on level $\ell-1$, and $K^n_{\ell-1}$ is that kernel composed with itself $n$ times. Recall again that according to \cref{thm:mlda} the finest sampler targets the exact posterior, so that $\tilde{\pi}_L = \pi_L$. Thus, the MLDA estimator in \cref{eq:mlda} approximates the following telescoping sum:
\begin{equation}
\mathbb E_{\pi_L}[Q_L] = \mathbb E_{\tilde{\pi}_0}[Q_0] + \sum_{\ell=1}^L\Big( \mathbb E_{\tilde{\pi}_\ell}[Q_\ell] - \mathbb E_{\tilde{\pi}_{\ell-1}}[Q_{\ell-1}]\Big),
\end{equation}
which is a small but crucial difference to the sum in Eq.~\cref{eqn:ML_Estimator} that forms the basis of MLMCMC \cite{dodwell_hierarchical_2015}.

The computational gains due to multilevel variance reduction remain. In fact, since the mixture densities $\tilde{\pi}_{\ell-1}$ are conditioned every $J_\ell$ steps on the next finer chain, they are even closer and thus, the variances of the correction terms in \cref{eq:mlda} will be further reduced compared to the variances of the estimates in \cref{eqn:ML_Telescoping}. The fixed subchain lengths $J_\ell$ and thus the numbers of samples $N_\ell$ on the coarser levels can then be chosen as usual in multilevel Monte Carlo approaches to minimise the total variance for a fixed computational budget, or to minimise the cost to achieve the smallest variance. We are not going to go into more depth with respect to this estimator in this paper, but refer to e.g. \cite{cliffe_multilevel_2011, dodwell_hierarchical_2015, giles_multilevel_2018} for detailed analyses of Multilevel (Markov Chain) Monte Carlo estimators.

\subsection{Adaptive Correction of the Approximate Posteriors to Improve Efficiency} \label{sec:adapt}

While the algorithm outlined in \cref{sec:mlda_convergence} does guarantee sampling from the exact posterior, there are situations where convergence can be prohibitively slow. When the coarse model approximations are poor, the second-stage acceptance probability can be low, and many proposals will be rejected. This will result in suboptimal acceptance rates, poor mixing and low effective sample sizes. The leftmost panel in \cref{fig:inflation_adaption} shows a contrived example where the approximate likelihood function (red isolines) is offset from the exact likelihood function (blue contours) and its scale, shape and orientation are incorrect.

One way to alleviate this problem is through {\em tempering}, where the variance in the likelihood function ${\mathbf \Sigma}_\epsilon$ on levels $\ell < L$ is inflated, resulting in a wider approximate posterior distribution. While this approach would allow the approximate posterior to encapsulate the exact posterior, it does not tackle the challenge in an intelligent fashion, and the inflation factor introduces an additional tuning parameter.

\begin{figure}[t]
    \centering
    \includegraphics[width=1\linewidth]{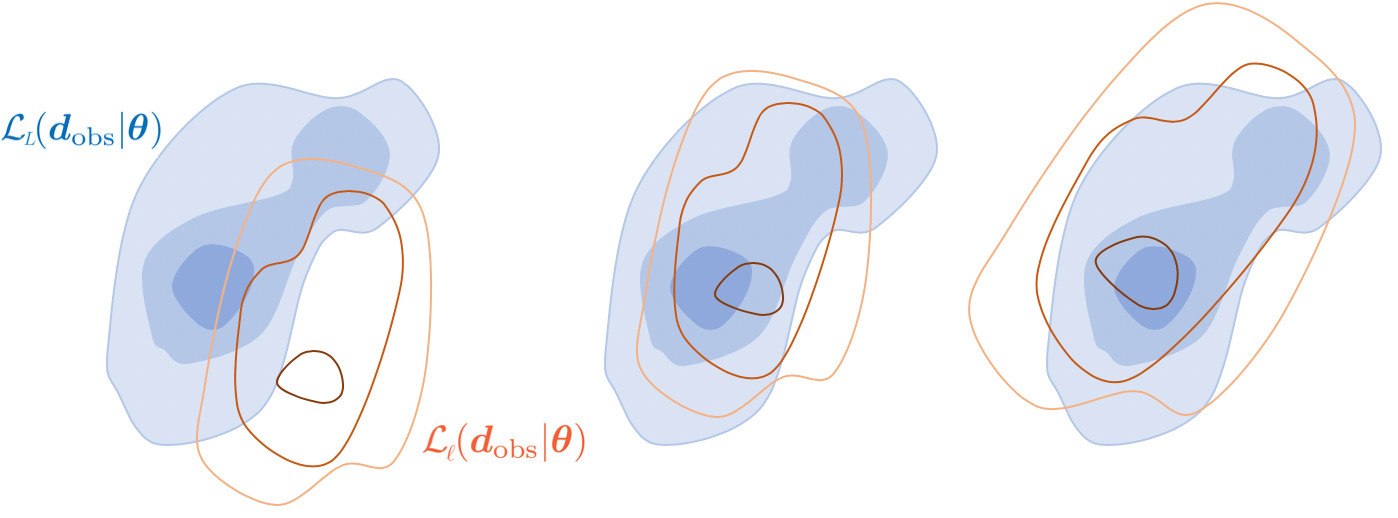}
    \caption{Effect of applying the Gaussian Adaptive Error Model (AEM). The first panel shows the initial state before adaptation, where the coarse likelihood function ($\mathcal{L}_{\ell}({\bf d}_{obs}|\theta)$, red isolines) approximates the fine likelihood function ($\mathcal{L}_{L}({\bf d}_{obs}|\theta)$, blue contours) poorly. The second panel shows the effect of adding the mean of the bias to the likelihood functional, resulting in an offset of the coarse model likelihood function. The third panel shows the effect of also adding the covariance of the bias to the likelihood functional, resulting in a scaling and rotation of the coarse likelihood function. Adapted from \cite{lykkegaard_accelerating_2021}.}
    \label{fig:inflation_adaption}
\end{figure}

In place of tempering, an enhanced Adaptive Error Model (AEM) can be employed to account for discrepancies between model levels. Let $\mathcal{F}_{\ell}$ denote the coarse forward map on level $\ell$ and $\mathcal{F}_{L}$ denote the forward map on the finest level~$L$. To obtain a better approximation of the data $d$ using $\mathcal{F}_{\ell}$, the two-level AEM suggested in \cite{cui_bayesian_2011} and analysed in \cite{cui_posteriori_2019,fox2020randomized} is extended here by adding a telescopic sum of the differences in the forward model output across all levels from $\ell$ to $L$:
\begin{equation} \label{eq:inverse_problem_bias}
d = \mathcal F_L(\theta) + \epsilon = \mathcal F_{\ell}(\theta) + \mathcal B_{\ell}(\theta) + \epsilon \quad \mbox{with} \quad
\mathcal B_{\ell}(\theta) := \sum_{k=\ell}^{L-1} \underbrace{\mathcal F_{k+1}(\theta) -  \mathcal F_{k}(\theta)}_{:=B_k(\theta)}
\end{equation}
denoting the bias on level $\ell$ at ${\mathbf \theta}$. The trick in the context of MLDA is that, since $\mathcal{B}_{\ell}$ is just a simple sum, the individual bias terms $B_k$ from pairs of adjacent model levels can be estimated independently, so that new information can be exploited each time \textit{any} set of adjacent levels are evaluated for the same parameter value $\theta$.

Approximating each individual bias term $B_k = \mathcal F_{k+1} -  \mathcal F_{k}$ with a multivariate Gaussian $B^*_k \sim \mathcal N(\mu_k, \mathbf \Sigma_k)$, the total bias $\mathcal{B}_\ell$ can be approximated by the Gaussian $\mathcal B^*_{\ell} \sim \mathcal N(\mu_{\mathcal{B}, \ell}, \mathbf \Sigma_{\mathcal{B}, \ell})$ with $\mu_{\mathcal{B}, \ell} = 
\sum_{k=\ell}^{L-1} \mu_k$ and $\mathbf \Sigma_{\mathcal{B}, \ell} = \sum_{k=\ell}^{L-1} \mathbf \Sigma_k$.

The bias-corrected likelihood function for level $\ell$ is then proportional to
\begin{equation} \label{eq:likelihood_adaptive} 
\mathcal{L}_{\ell}({\mathbf d} | \theta) = \exp \left( -\frac{1}{2} (\mathcal{F}_{\ell}(\theta) + {\mathbf \mu}_{\mathcal{B}, \ell} - {\mathbf d})^T ({\mathbf \Sigma}_{\mathcal{B}, \ell} + {\mathbf \Sigma}_e)^{-1} (\mathcal{F}_{\ell}(\theta) + {\mathbf \mu}_{\mathcal{B}, \ell} - {\mathbf d}) \right).
\end{equation}
The \emph{Approximation} Error Model, suggested by \cite{kaipio_statistical_2007}, is constructed offline, by sampling from the prior distribution before running the MCMC; We simply sample $N$ parameter sets from the prior and compute the sample moments according to
\begin{equation} 
    {\bf \mu}_k    = \frac{1}{N} \sum_{i=1}^{N} B_k(\theta^{(i)}) \quad \mbox{and} \quad    {\bf \Sigma}_{k} = \frac{1}{N-1} \sum_{i=1}^{N} (B_k(\theta^{(i)}) - {\bf \mu}_k)(B_k(\theta^{(i)}) - {\bf \mu}_k)^{T}.
\end{equation}
However, this approach requires significant investment prior to sampling, and may result in a suboptimal error model, since the bias in the posterior distribution  is very different from the bias in the prior when the data is informative. Instead, as suggested in \cite{cui_bayesian_2011}, an estimate for $B_k$ can be constructed iteratively during sampling, using the following recursive formulae 
for sample means and sample covariances~\cite{haario_adaptive_2001}:
\begin{equation}
   {\mu}_{k,i+1}    = \frac{1}{i+1} \Big( i {\mu}_{k,i} + B_k(\theta^{i+1}) \Big) \quad \mbox{and}
\end{equation}
\begin{equation}
    {\bf \Sigma}_{k,i+1} = \frac{i-1}{i} {\bf \Sigma}_{k,i} + \frac{1}{i} \Big( i{\mu}_{k,i}\: { \mu}_{k,i}^T - (i+1) {\mu}_{k,i+1}\: { \mu}_{k,i+1}^T + B_k(\theta^{i+1}) \: B_k(\theta^{i+1})^T \Big).
\end{equation}
While this approach in theory results in a MCMC algorithm that is not Markov, the recursively constructed sample moments converge as sampling proceeds and hence the approach exhibits \textit{diminishing adaptation} and \emph{bounded convergence} which is sufficient to ensure ergodicity for adaptive MCMC schemes, \cite{roberts_coupling_2007, roberts_examples_2009}. As shown in \cite{cui_posteriori_2019}, it is also possible to construct a \textit{state-dependent} AEM, where the coarse samples are corrected only according to the bias of the state of the MCMC, rather than the mean of the bias. This approach, however, may require a different form of the multilevel acceptance probability (Eq.~\cref{eq:beta_mlda}), which we have not yet established, as discussed in \cref{sec:mlda_convergence}. We remark that while the simple Gaussian error model described here does suffer from a limited expressiveness, it is robust. Any coarse-level bias that is nonlinear in the model parameters will be absorbed by the respective covariance term, which will allow the coarse levels to sample ``broader'' and certainly encapsulate the true posterior. The general bias-modelling framework described by Eq.~\cref{eq:inverse_problem_bias} allows for the bias terms to be modelled by any functions of the model parameters, including Gaussian processes, artificial neural networks, polynomial chaos expansions, etc., as long as they are either constructed \textit{a priori} or exhibit diminishing adaptation and bounded convergence. However, the Gaussian model proposed here requires does not require any tuning or caching of the bias history, and is both computationally cheap and numerically robust. Hence, unless a particular problem strongly favours a different bias modelling approach, we recommend the Gaussian model described above.

\section{Examples}
\label{sec:examp}

In this section, we consider three inverse problems which demonstrate the efficiency gains obtained by using MLDA, as well as by the extensions outlined above. The algorithm has been included in the free and open source probabilistic programming library \texttt{PyMC3}\footnote{\href{https://docs.pymc.io/en/v3/index.html}{https://docs.pymc.io/en/v3/index.html}} as the \texttt{MLDA} step method since version 3.10.0, and the examples below were all completed using this implementation.

\subsection{Gravitational Survey}\label{sec:gravity}

In this example, we consider a 2-dimensional gravity surveying problem, adapted from the 1-dimensional problem presented in \cite{Hans10}. Our aim is to recover an unknown two-dimensional mass density distribution $f({\bf t})$ at a known depth $d$ below the surface from measurements $g({\bf s})$ of the vertical component of the gravitational field at the surface. The contribution to $g({\bf s})$ from infinitesimally small areas of the subsurface mass distribution are given by:
\begin{equation}
    \text{d}g({\bf s}) = \frac{\sin \theta}{r^2} f({\bf t}) \: \text{d}{\bf t}
\end{equation}
where $\theta$ is the angle between the vertical plane and a straight line between two points ${\bf t}$ and ${\bf s}$, and $r = \Vert {\bf s} - {\bf t} \Vert_2$ is the Eucledian distance between the points. We exploit that $\sin \theta = d/r$, so that
\begin{equation}
    \frac{\sin \theta}{r^2} f({\bf t}) \: \text{d}{\bf t} = \frac{d}{r^3} f({\bf t}) \: \text{d}{\bf t} = \frac{d}{ \Vert {\bf s} - {\bf t} \Vert_2^3} f({\bf t}) \: \text{d}{\bf t}
\end{equation}
This yields the integral equation
\begin{equation}
    g({\bf s}) = \int \! \int_T \frac{d}{ \Vert {\bf s} - {\bf t}\Vert_2^3} f({\bf t}) \: \text{d}{\bf t}
\end{equation}
where $T = [0,1]^2$ is the domain of the function $f({\bf t})$. This constitutes our forward model. We solve the integral numerically using midpoint quadrature. For simplicity, we use $m$ quadrature points along each dimension, so that in discrete form our forward model becomes
\begin{equation}
    g({\bf s}_i) 
    = \sum_{l=1}^{m} \omega_l \sum_{k=1}^{m} \omega_k \frac{d}{ \Vert {\bf s}_i - {\bf t}_{k,l} \Vert_2^3} \hat{f}({\bf t}_{k,l}) 
    = \sum_{j=1}^{m^2} \omega_j \frac{d}{ \Vert {\bf s}_i - {\bf t}_j \Vert_2^3} \hat{f}({\bf t}_j)
\end{equation}
where $\omega_j = 1/m^2$ are the quadrature weights, $\hat{f}({\bf t}_j)$ is the approximate subsurface mass at the quadrature points ${\bf t}_j, \ j = 1, \dots, m^2$, and  $g({\bf s}_i)$ is the surface measurement at the collocation point ${\bf s}_i, \ i = 1, \dots, n^2$. Hence, when $n > m$, we are dealing with an overdetermined problem and vice versa. This can be expressed as a linear system $\mathbf{Ax = b}$, where
\begin{equation}
    a_{ij} = \omega_j \frac{d}{ \Vert {\bf s}_i - {\bf t}_j \Vert_2^3}, \quad x_j = \hat{f}({\bf t}_j), \quad b_i = g({\bf s}_i).
\end{equation}
Due to the ill-posedness of the underlying, continuous inverse problem, the matrix $\mathbf{A}$ is very ill-conditioned,  which entails numerical instability and spurious, often oscillatory, naive solutions for noisy right hand sides. A problem of this type is traditionally solved by way of \textit{regularisation} such as Tikhonov regularisation or Truncated Singular Value Decomposition (TSVD), but it can also be handled in a more natural and elegant fashion as a Bayesian inverse problem.

For the exerimental set-up, a ``true'' mass density distribution $f(t)$ was assigned on $T$ at a depth of $d = 0.1$ (\cref{fig:gravity_true}, left panel). The modelled signal was then discretised with $m = n = 100$ and
perturbed with white noise with standard deviation $\sigma_{\epsilon} = 0.1$ (\cref{fig:gravity_true}, right panel) to be used as synthetic data in the numerical experiment.
\begin{figure}[t]
\centering
\includegraphics[width = 0.43\linewidth]{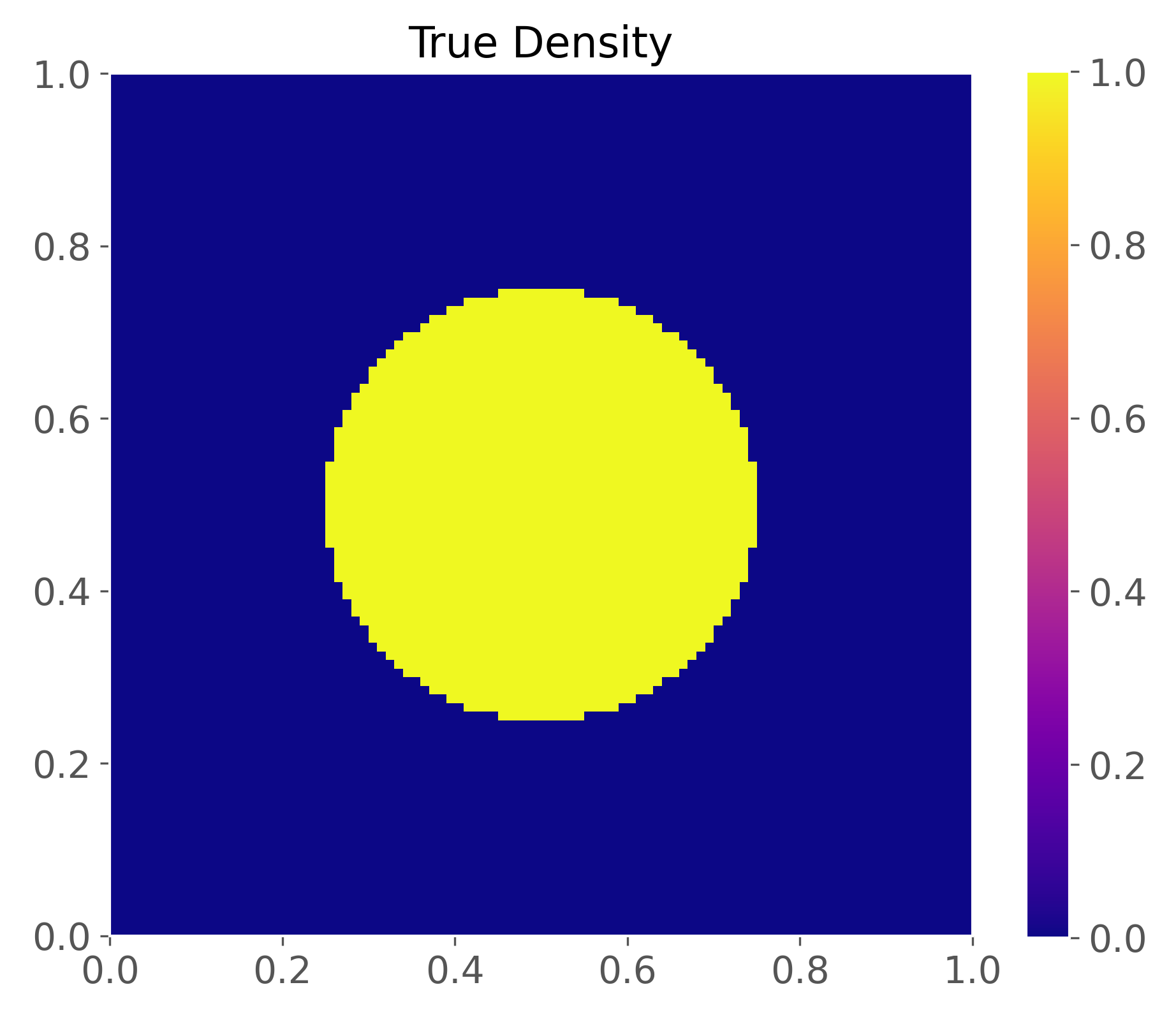}
\hspace{1cm}
\includegraphics[width = 0.43\linewidth]{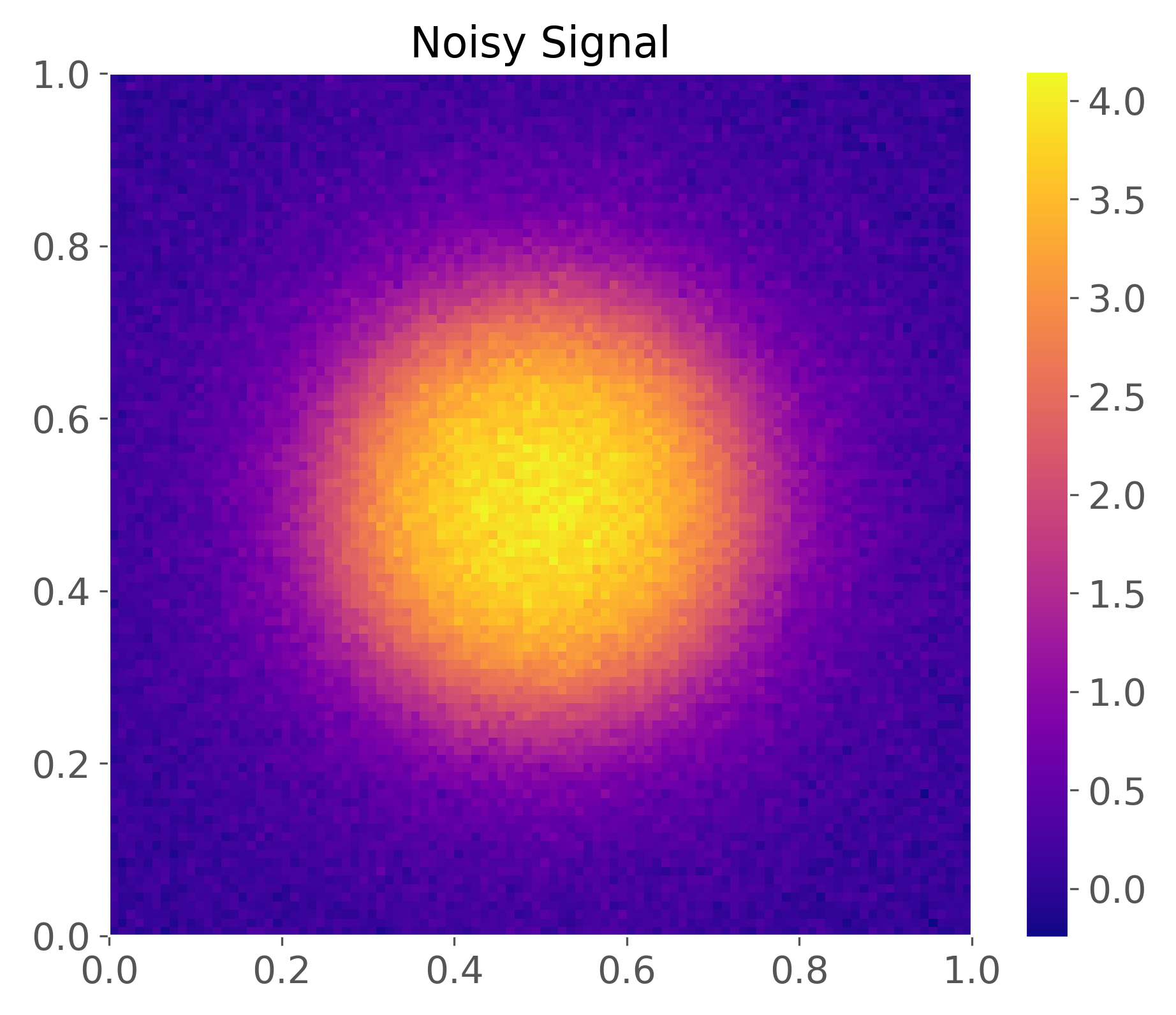}
\caption{(Left) The ``true'' mass density $f(t)$ and (right) the noisy signal at $d=0.1$, with $\sigma_{\epsilon} = 0.1$.}
 \label{fig:gravity_true}
\end{figure}

The unknown mass density distribution was modelled as a Gaussian Random Process with a Matérn 3/2 covariance kernel \cite{Ras06}:
\begin{equation}
    C_{3/2}({x}, {y}) = \sigma^2 \left( 1 + \frac{\sqrt{3} \Vert {x}-{y} \Vert_2 }{\lambda} \right) \exp \left( - \frac{\sqrt{3} \Vert {x}-{y} \Vert_2 }{\lambda} \right), \quad \mbox{for} \quad {x}, {y} \in D,
\end{equation}
where $\lambda$ is the covariance length scale and $\sigma^2$ is the variance.
The random 
field was parametrised using a truncated Karhunen-Lo\`eve (KL) expansion 
of $f(t)$, i.e. an expansion in terms of a finite set of independent, standard Gaussian random variables $\theta_i \sim \mathcal{N}(0,1)$, $i=1,\ldots,R$, given by
\begin{equation}
f(t,\omega) = \sum_{i=1}^{R} \sqrt{\mu_i} \phi_i({t})\theta_i(\omega).
\end{equation}
Here, $\{\mu_i\}_{i \in \mathbb N}$ are the sequence of strictly decreasing real, positive eigenvalues, and $\{\phi_i\}_{i\in \mathbb N}$ the corresponding $L^2$-orthonormal eigenfunctions of the covariance operator with kernel $C_{3/2}(x,y)$.

A model hierarchy consisting of two model levels, with $m = 100$ and $m = 20$ respectively, was created. A Matern 3/2 random process with $l = 0.2$ and $\sigma^2 = 1$ was initialised on the fine model level and  parametrised using KL decomposition, which was then truncated to encompass its $R=32$ highest energy eigenmodes. It was then projected to the coarse model space (\cref{fig:gravity_model}).

\begin{figure}[t]
\centering
\includegraphics[width = 0.43\linewidth]{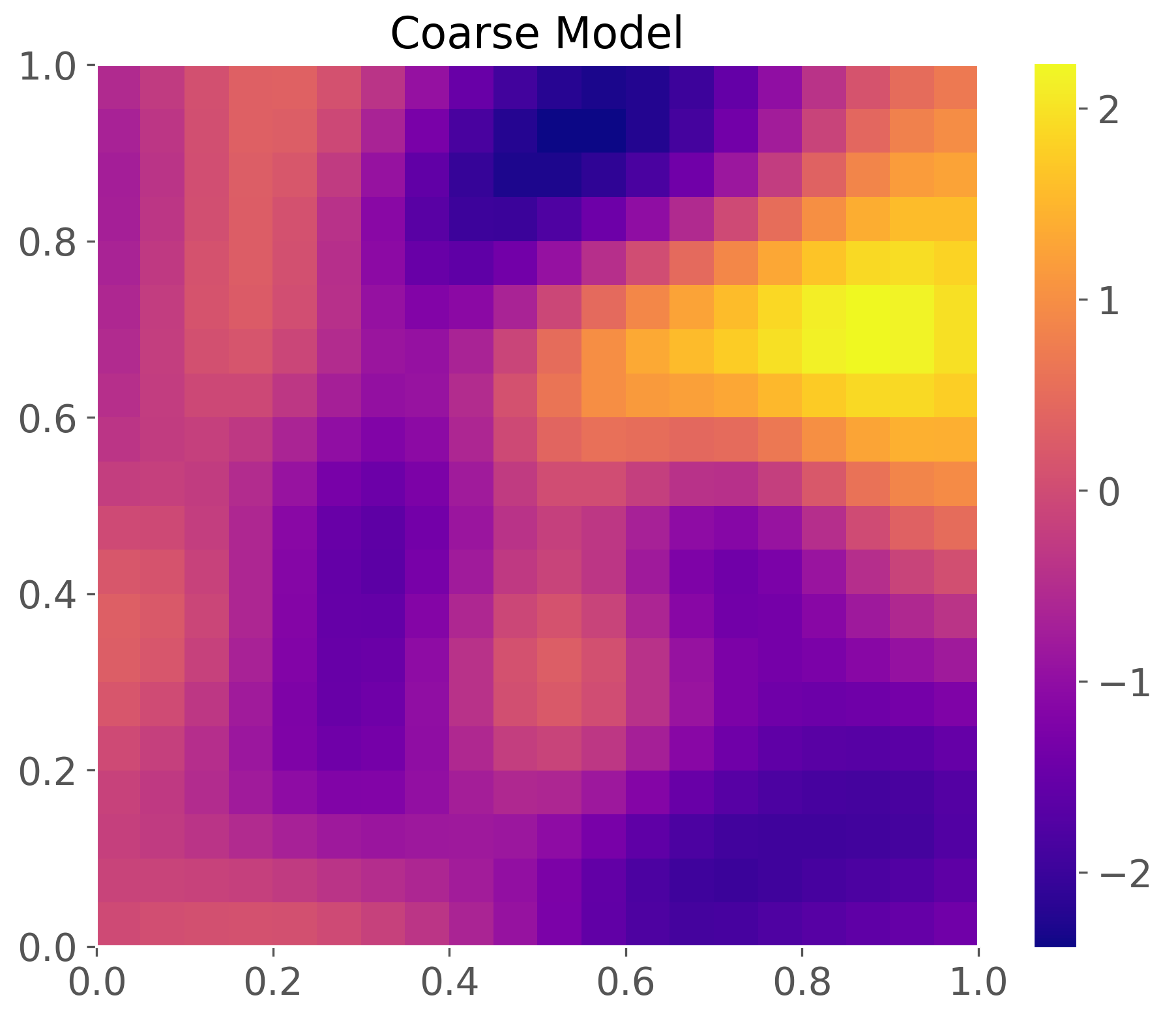}
\hspace{1cm}
\includegraphics[width = 0.43\linewidth]{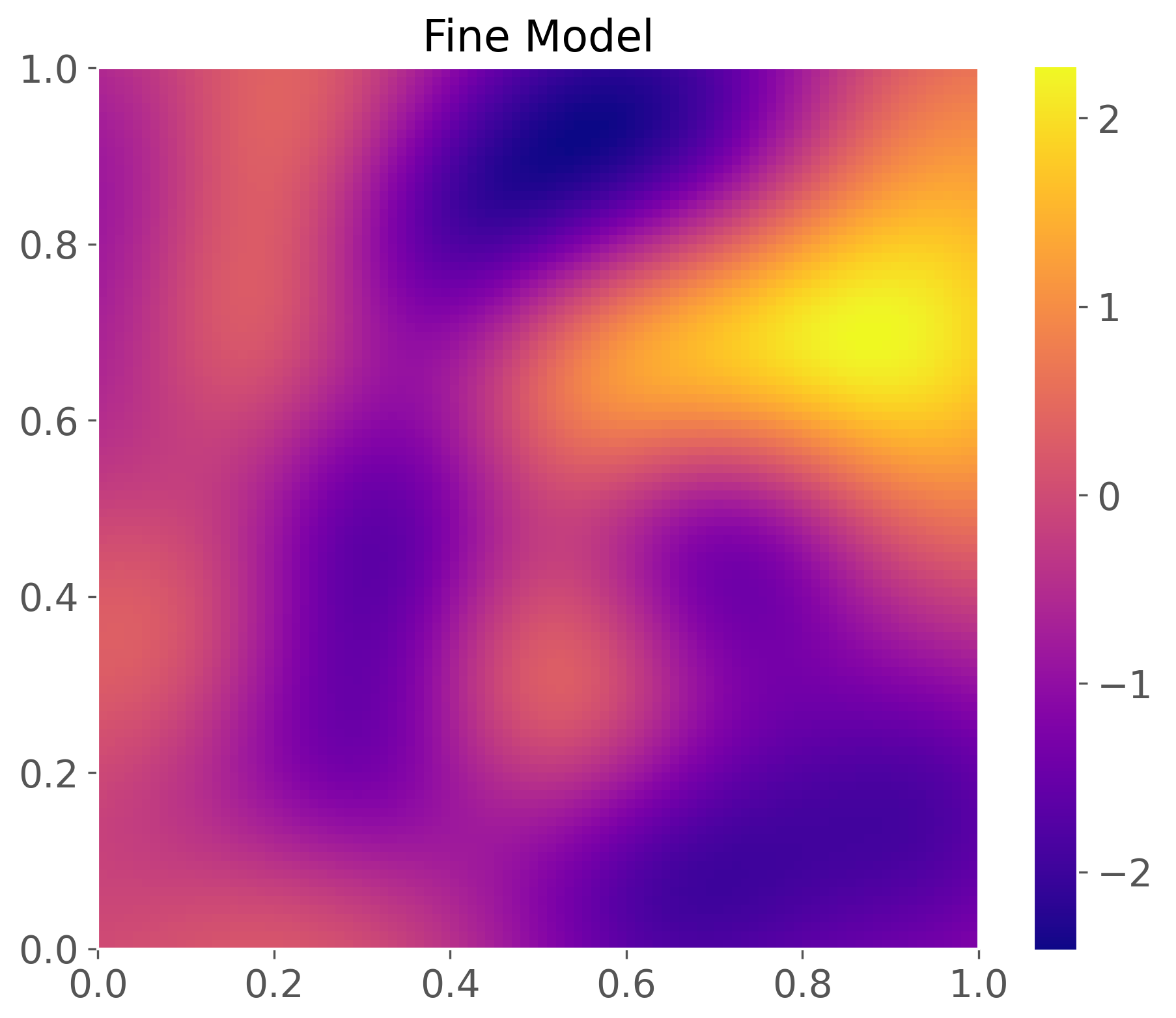}
\caption{Random realisations of the Matérn 3/2 random process prior, used to model the unknown mass density for the coarse model with $m = 20$ (left) and the fine model with $m = 100$ (right).}
\label{fig:gravity_model}
\end{figure}

Thus, the prior distribution of the model parameters $(\theta_i)_{i=1}^{R}$ is $\mathcal{N}(0, I_R)$. To sample from the posterior distribution of these parameters and thus to estimate the posterior mean conditioned on the synthetic data, we used the TLDA sampler with a Random Walk Metropolis Hastings (RWMH) sampler on the coarse level. We ran 2 independent chains, each with 20000 draws, a burn-in of 5000 and a subchain length on the coarse level of 10. We also ran 2 chains using a single level RWMH sampler on the fine level with otherwise identical settings, but with no subchains. Each chain was initialised at the MAP (Maximum a Posteriori) point.

While RWMH converged to the same parameter estimates as MLDA, RWMH exhibited inferior mixing (\cref{fig:gravity_traces}) and fewer effective samples per second (\cref{fig:gravity_performance}), particularly for the higher KL coefficients.
\begin{figure}[t]
\includegraphics[width = 1\linewidth]{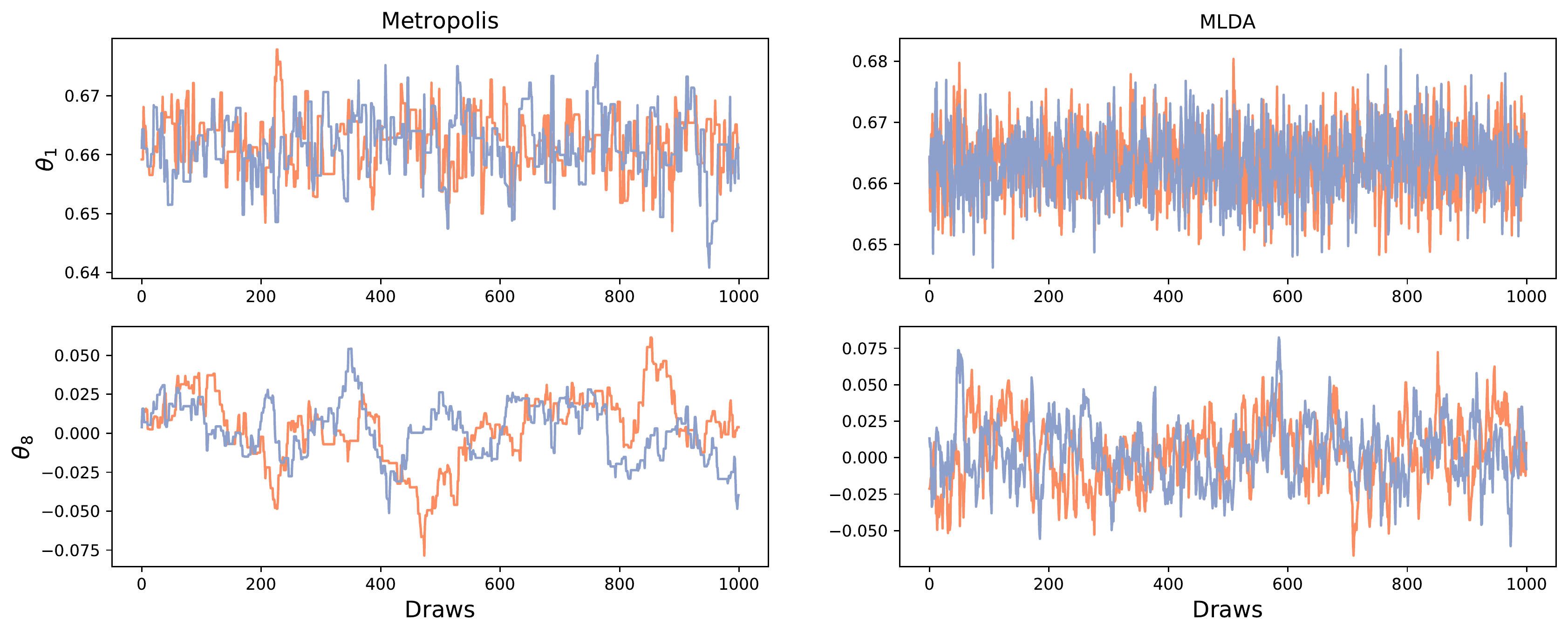}
\caption{Traces of $\theta_1$ (top row) and $\theta_8$, for RWMH (left column) and MLDA (right column), respectively. Different colors represent the independent chains.}
\label{fig:gravity_traces}
\end{figure}
\begin{figure}[t]
\centering
\includegraphics[width = 0.7\linewidth]{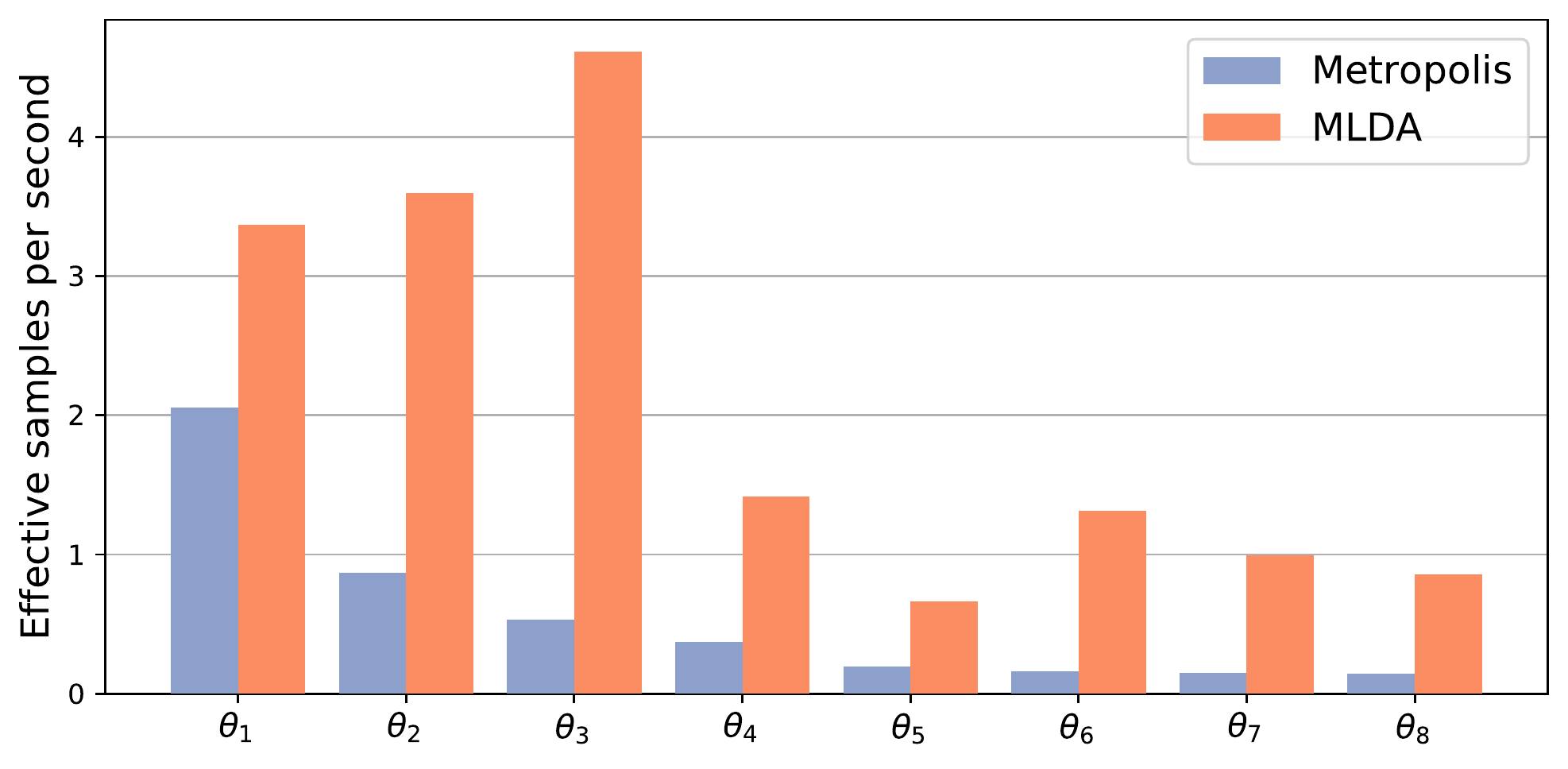}
\caption{Algorithmic performance measured in ES/s (effective samples per second), for the eight highest energy KL coefficients $\theta_k, k = 1, \dots, 8$, for both RWMH (blue) and MLDA (red).}
\label{fig:gravity_performance}
\end{figure}

\subsection{Predator-Prey Model}

The Lotka-Volterra model describes the interaction between populations of prey ($N$) and predators ($P$) over time \cite{rockwood_introduction_2015}. Their interaction is described by the system of nonlinear, first order, ordinary differental equations (ODEs) 
\begin{equation}\label{eqn:pp_model}
\frac{dN}{dt} = aN - bNP \quad \mbox{and} \quad \frac{dP}{dt} = cNP - dP, \quad \mbox{for} \: t > 0.
\end{equation}
The model outputs are fully described by the parameters
$$\theta = \{N_0, P_0, a, b, c, d\},$$
which include the initial densities of prey and predators at time $t=0$, and ecological parameters $a, b, c, d$, where broadly $a$ is the birth rate of the prey, $b$ is the encounter rate between prey and predators, $c$ is the growth rate for the predators and $d$ is the death rate of the predators. For further details on their physical interpretation see for example \cite{bacaer_short_2011}.
\begin{figure}[t]
    \centering
    \includegraphics[width=1.0\linewidth]{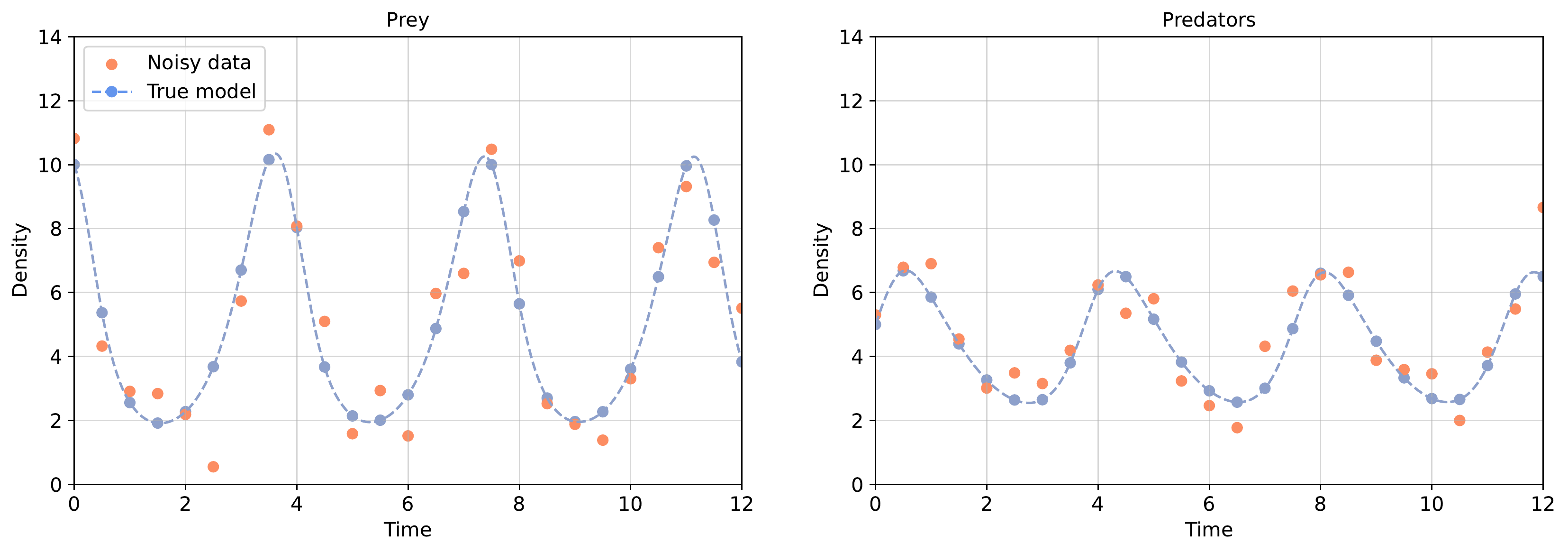}
    \caption{\label{fig:pp_data}The true (blue) and measured (red) densities of prey (left) and predators (right).}
\end{figure}

In this example, we wish to infer the distribution of $\theta$, given noisy observations of prey and predator densities at discrete time intervals, i.e. $N(t^\star)$ and $P(t^\star)$ for $t^\star \in \mathcal T$, where $\mathcal T = [0,12]$ is the domain. The observations are again synthetically generated by solving Eq.~\cref{eqn:pp_model} with the ``true'' parameters
$$
\theta^\star = \{10.0, 5.0, 3.0, 0.7, 0.2, 1.0\}
$$
and perturbing the calculated values $N(t^\star)$ and $P(t^\star)$  with independent Gaussian noise $\epsilon \sim \mathcal N(0, 1)$ (\cref{fig:pp_data}). Our aim is to predict the mean density of predators $\mathbb E(P)$ over the same period.

The solutions of the ODE system in Eq.~\cref{eqn:pp_model} can be approximated by a suitable numerical integration scheme. We use an explicit, adaptive Runge-Kutta method of order 5(4) \cite{strogatz_nonlinear_2007}. For the finest level $\ell=2$, we integrate over the entire time domain $\mathcal T_2 = [0,12]$ and use the entire dataset to compute the likelihood function, while for the coarse levels, we stop integration early, so that $\mathcal T_{1} = [0,8]$ and $\mathcal T_{0} = [0,4]$, and use only the corresponding subsets of the data to compute the likelihood functions.

We assume that we possess some prior knowledge about the parameters, and use informed priors $N_0 \sim \mathcal N(10.8, 1)$, $P_0 \sim \mathcal N(5.3, 1)$, $a \sim \mathcal N(2.5, 0.5)$, $b \sim \text{Inv-Gamma}(1.0, 0.5)$, $c \sim \text{Inv-Gamma}(1.0, 0.5)$ and $d \sim \mathcal N(1.2, 0.3)$.

To demonstrate the multilevel variance reduction feature, we ran the MLDA sampler with randomisation of the subchain length as described in \cref{sec:VarianceReduction} and then compared the (multilevel) MLDA estimator in Eq.~\cref{eq:mlda}, which uses both the coarse and fine samples, with a standard MCMC estimator based only on the samples produced by MLDA on the fine level. In both cases, we used the three--level model hierarchy as described above and employed the Differential Evolution Markov Chain (DE-MC\textsubscript{Z}) proposal \cite{ter_braak_differential_2008} on the coarsest level. The coarsest level proposal kernel was automatically tuned during burn-in to achieve an acceptance rate between 0.2 and 0.5. The subchain lengths of $J_2 = J_{1} = 10$ were chosen to balance the variances of the two contributions to the multilevel estimator (Eq.~\cref{eq:mlda}), as for MLMC and MLMCMC.

\cref{fig:pp_se} shows the development of the total sampling error as the sampling progresses, for the sampler with and without variance reduction. Employing variance reduction clearly leads to a lower sampling error than the standard approach. \cref{fig:pp_samples} shows the true prey and predator densities along with samples from the posterior distribution, demonstrating that the true model is encapsulated by the posterior samples, as desired.
\begin{figure}[t]
    \centering
    \includegraphics[width=0.6\linewidth]{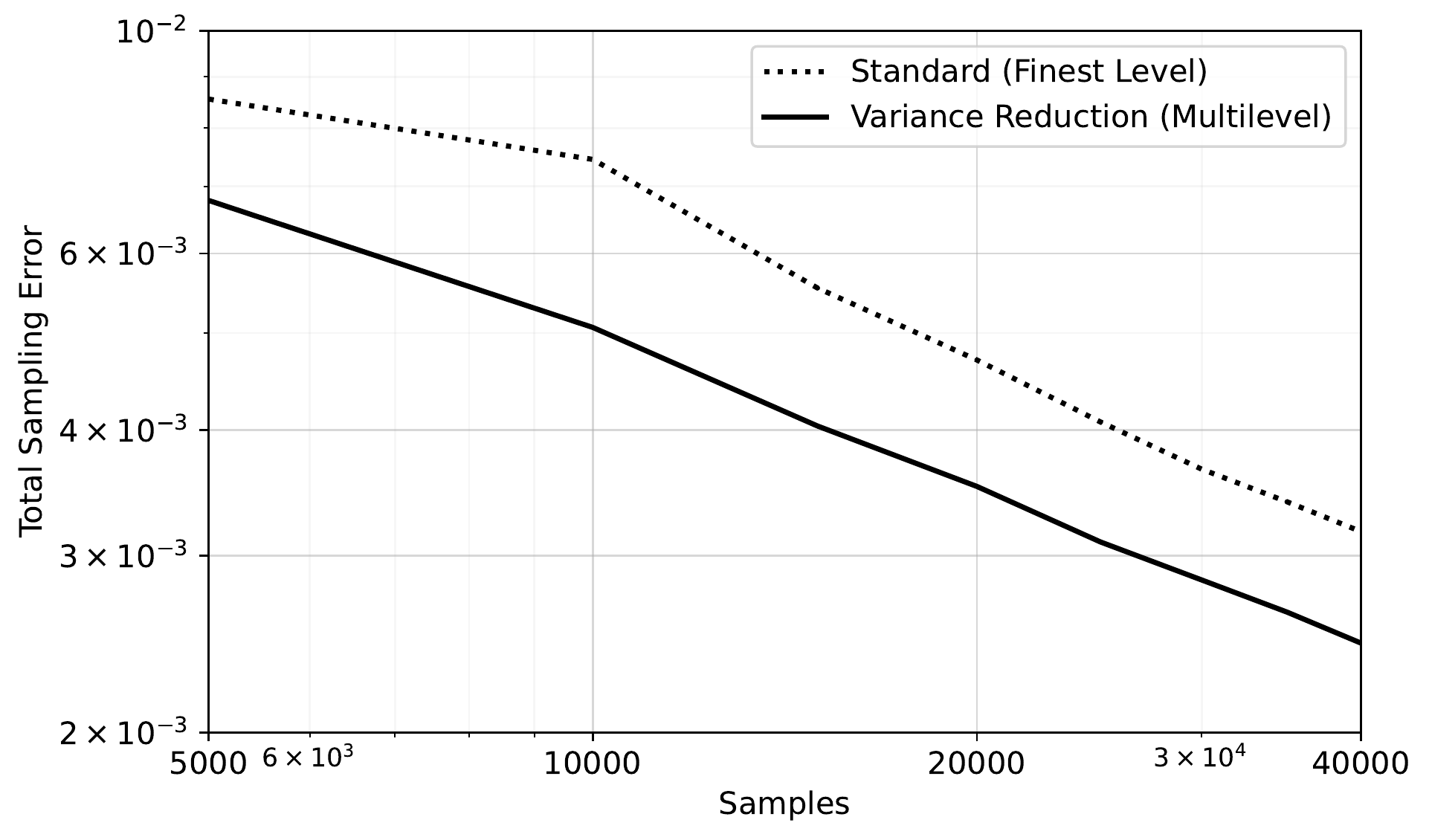}
    \caption{\label{fig:pp_se}Development of the total sampling error as sampling progresses for the sampler with (solid) and without (dashed) variance reduction.}
\end{figure}
\begin{figure}[t]
    \centering
    \includegraphics[width=1.0\linewidth]{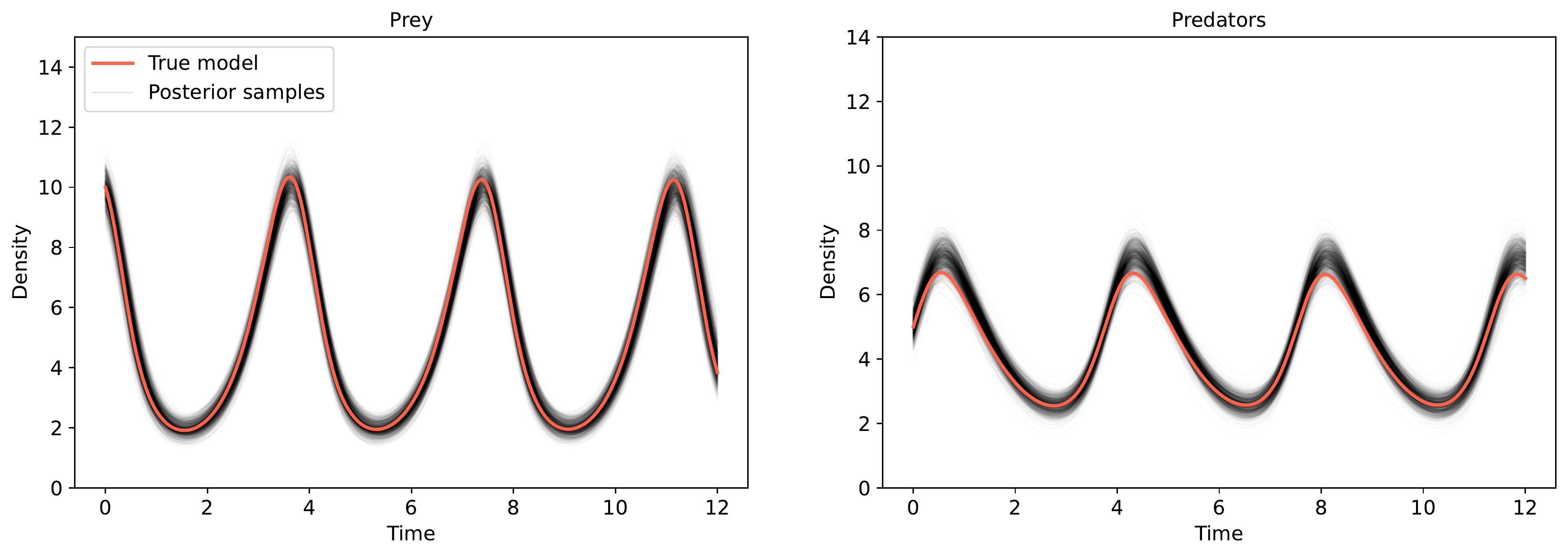}
    \caption{\label{fig:pp_samples}True model (red) and posterior samples (black).}
\end{figure}

\subsection{Subsurface Flow}\label{sec:rocks}

In this example, a simple model problem arising in subsurface flow modelling is considered. Probabilistic uncertainty quantification is of interest in various situations, for example in risk assessment of radioactive waste repositories. Moreover, this simple PDE model is often used as a benchmark for MCMC algorithms in the applied mathematics literature \cite{marzouk_stochastic_2007, marzouk_dimensionality_2009, dodwell_hierarchical_2015, conrad_accelerating_2016, conrad_parallel_2017, beskos_geometric_2017}. The classical equations which govern steady-state single-phase subsurface flow in a confined aquifer are Darcy's law coupled with an incompressibility constraint
\begin{equation}\label{eqn:fullDarcyEquations}
w + k\nabla p = g \quad \mbox{and} \quad \nabla \cdot w = 0, \quad \mbox{in} \quad D \subset \mathbb{R}^d
\end{equation}
for $d = 1,2$ or $3$, subject to suitable boundary conditions. Here $p$ denotes the hydraulic head of the fluid, $k$ the permeability tensor, $w$ the flux and $g$ is the source term.

A typical approach to treat the inherent uncertainty in this problem is to model the permeability as a random field $k = k(x,\omega)$ on $D \times \Omega$, for some probability space $(\Omega, \mathcal A, \mathbb P)$. Therefore, Eq.~\cref{eqn:fullDarcyEquations} can be written as the following PDE with random coefficients:
\begin{equation}\label{eq:spde}
-\nabla \cdot k(x,\omega)\nabla p(x,\omega) = f(x), \quad \mbox{for all} \quad x \in D,
\end{equation}
where $f:=-\nabla \cdot g$.
As a synthetic example, consider the domain $D := [0,1]^2$ with $f\equiv 0$ and deterministic boundary conditions
\begin{equation}
p|_{x_1=0} = 0, \quad p\vert_{x_1=1} = 1 \quad \mbox{and} \quad \partial_n p \vert_{x_2=0} = \partial_n p\vert_{x_2=1} = 0.
\end{equation}
A widely used model for the prior distribution of the permeability in hydrology is a log-Gaussian random field \cite{dodwell_hierarchical_2015, constantine_accelerating_2016, conrad_accelerating_2016, beskos_geometric_2017, lan_adaptive_2019}, characterised by the mean of $\log k$, here
chosen to be $0$, and by its covariance function, here chosen to be\vspace{-1ex}
\begin{equation}\label{eqn:covariance}
C({x},{y}) := \sigma^2 \exp\left(-\frac{\|{x} - {y}\|^2_2}{2\lambda^2}\right), \quad \mbox{for} \quad {x}, {y} \in D,
\end{equation}
with $\sigma = 2$ and $\lambda = 0.1$. Again, the log-Gaussian random 
field is parametrised using a truncated Karhunen-Lo\`eve (KL) expansion 
of $\log k$, i.e., an expansion in terms of a finite set of independent, standard Gaussian random variables $\theta_i \sim \mathcal{N}(0,1)$, $i=1,\ldots,R$, given by
\begin{equation}
\log k(x,\omega) = \sum_{i=1}^R \sqrt{\mu_i} \phi_i({x})\theta_i(\omega).
\end{equation}
Again, $\{\mu_i\}_{i \in \mathbb N}$ are the sequence of strictly decreasing real, positive eigenvalues, and $\{\phi_i\}_{i\in \mathbb N}$ the corresponding $L^2$-orthonormal eigenfunctions of the covariance operator with kernel $C(x,y)$.
Thus, the prior distribution on the parameter
$\theta = (\theta_i)_{i=1}^R$ in the stochastic PDE problem (Eq.~\cref{eq:spde}) is $\mathcal{N}(0,I_R)$. In this example we chose $R = 64$.

The aim is to infer the posterior distribution of $\theta$, conditioned on  measurements of $p$ at $M=25$ discrete locations $x^j \in D$, $j=1,\ldots,M$, stored in the vector ${d}_{obs} \in \mathbb R^{M}$. Thus, the forward operator is $\mathcal{F}:\mathbb{R}^R \to \mathbb{R}^M$ with $\mathcal{F}_j(\theta_\omega) = p(x^j,\omega)$.
\begin{figure}[t]
\centering
\includegraphics[width = 0.43\linewidth]{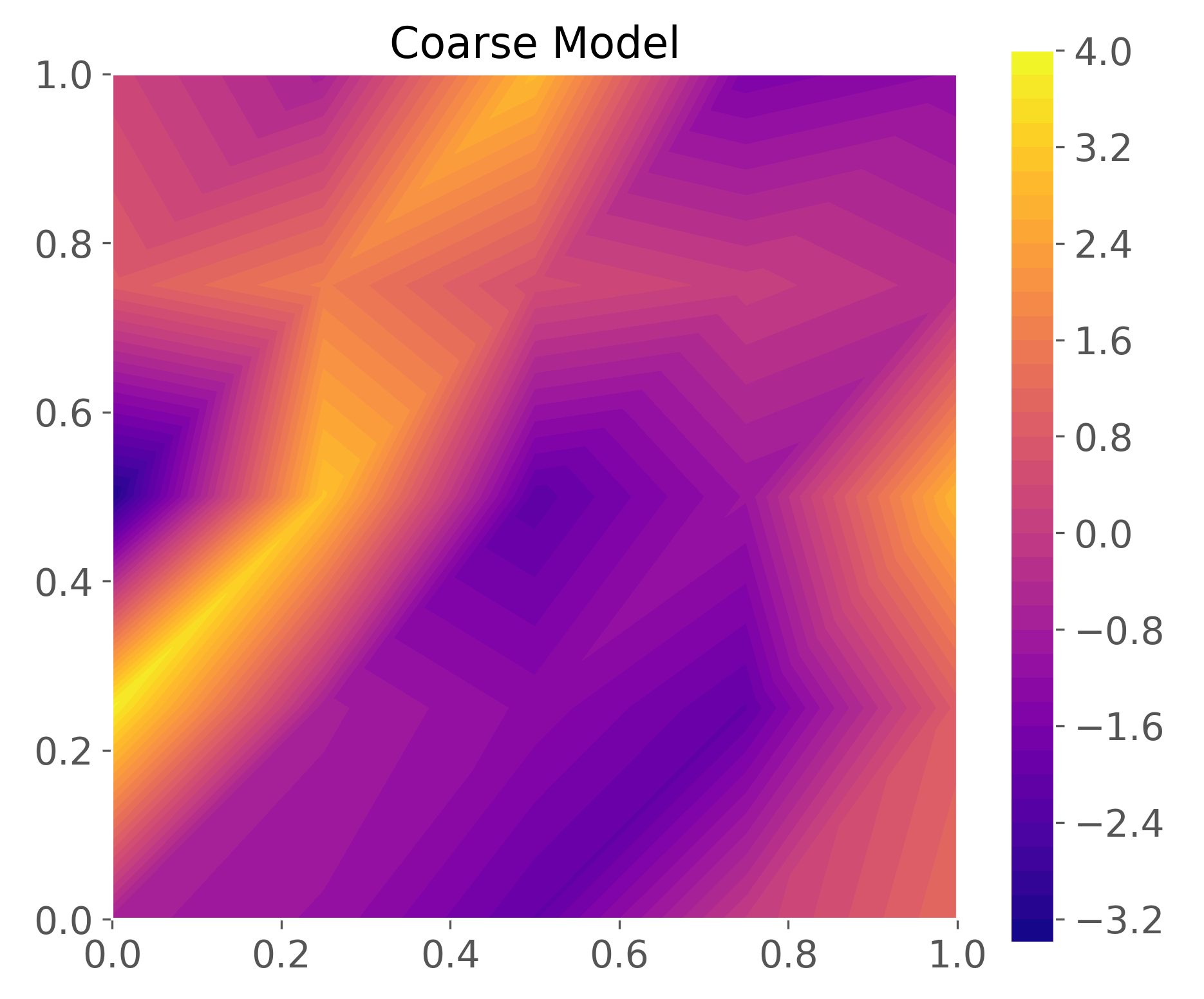}
\hspace{1cm}
\includegraphics[width = 0.43\linewidth]{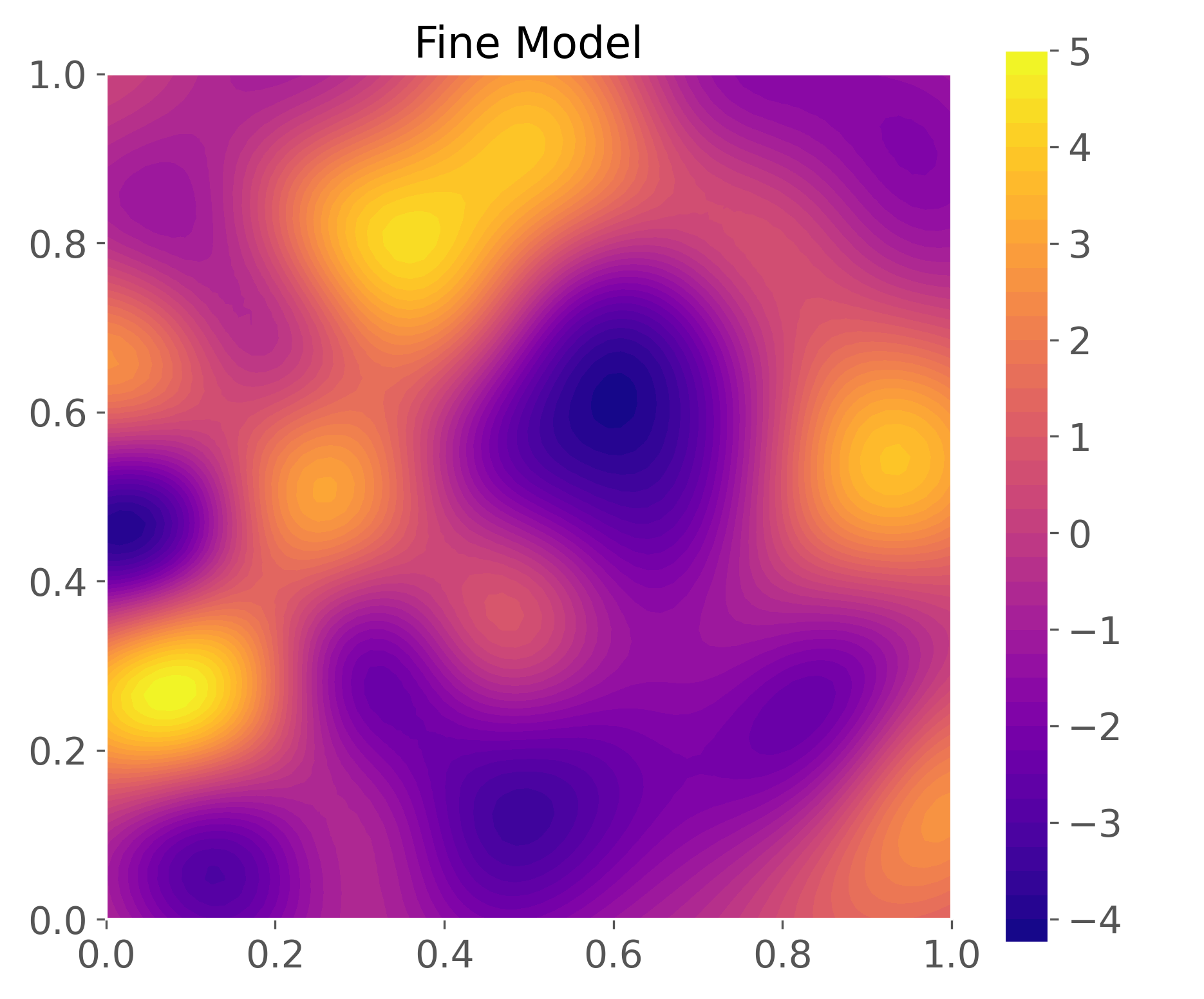}
\caption{True log-conductivity field of the coarsest model with $m_0$ grid points (left) and the finest model with $m_2$ grid points (right).}
\label{fig:gw_model}
\end{figure}

All finite element (FE) calculations were carried out with \texttt{FEniCS} \cite{langtangen_solving_2017}, using piecewise linear FEs on a uniform triangular mesh. The coarsest mesh $\mathcal T_0$ consisted of $m_0 = 5$ grid points in each direction, while subsequent levels were constructed by two steps of uniform refinement of $\mathcal T_0$, leading to $m_\ell = 4^\ell(m_0 -1) + 1$ grid points in each direction on the three grids $\mathcal{T}_\ell$, $\ell = 0, 1, 2$ (\cref{fig:gw_model}). 

To demonstrate the excellent performance of MLDA with the AEM, synthetic data was generated by drawing a sample from the prior distribution and solving (Eq.~\cref{eq:spde}) with the resulting realisation of $k$ on $\mathcal T_2$. To construct $d_{obs}$, the computed discrete hydraulic head values at $(x^j)_{j=1}^M$ were then perturbed by independent Gaussian noise, i.e. by a sample $\epsilon^* \sim \mathcal N(0, \Sigma_\epsilon)$ with $\Sigma_\epsilon = 0.01^2 I_M$.

To compare the ``vanilla'' MLDA approach to the AEM-enhanced version, we sampled the same model using identical sampling parameters, with and without AEM activated. For each approach, we sampled two independent chains, each initialised at a random point from the prior. For each chain, we drew 20000 samples plus a burn-in of 5000. We used subchain lengths $J_0 = J_1 =5$, since that produced the best trade-off between computation time and effective sample size for MLDA with the AEM. Note that the cost of computing the subchains on the coarser levels only leads to about a 50\% increase in the total cost for drawing a sample on level $L$. The DE-MC\textsubscript{Z} proposal \cite{ter_braak_differential_2008} was employed on the coarsest level with automatic step-size tuning during burnin to achieve an acceptance rate between 0.2 and 0.5.
\begin{figure}[t]
\centering
\includegraphics[width = 0.45\linewidth]{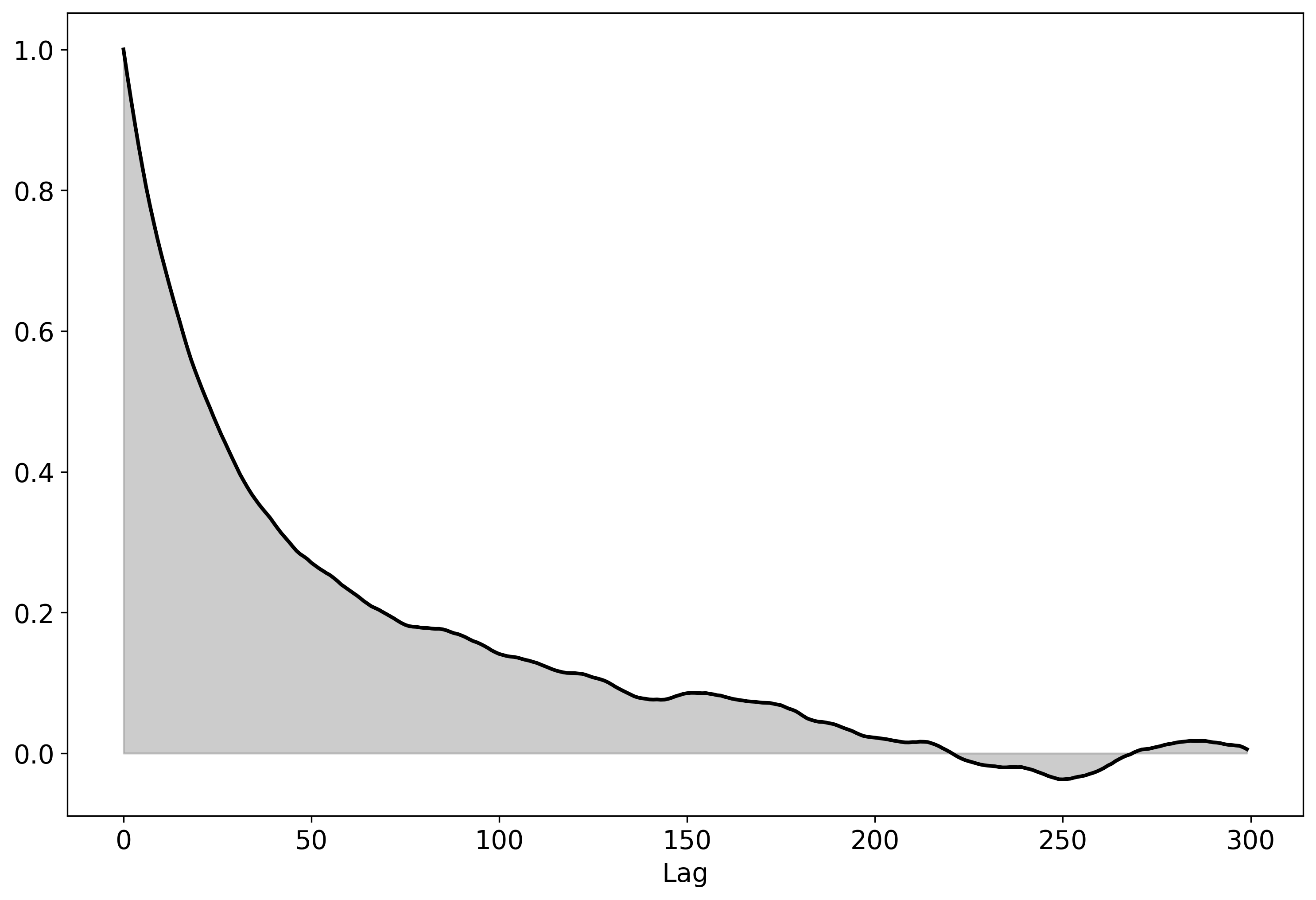}
\hspace{1cm}
\includegraphics[width = 0.45\linewidth]{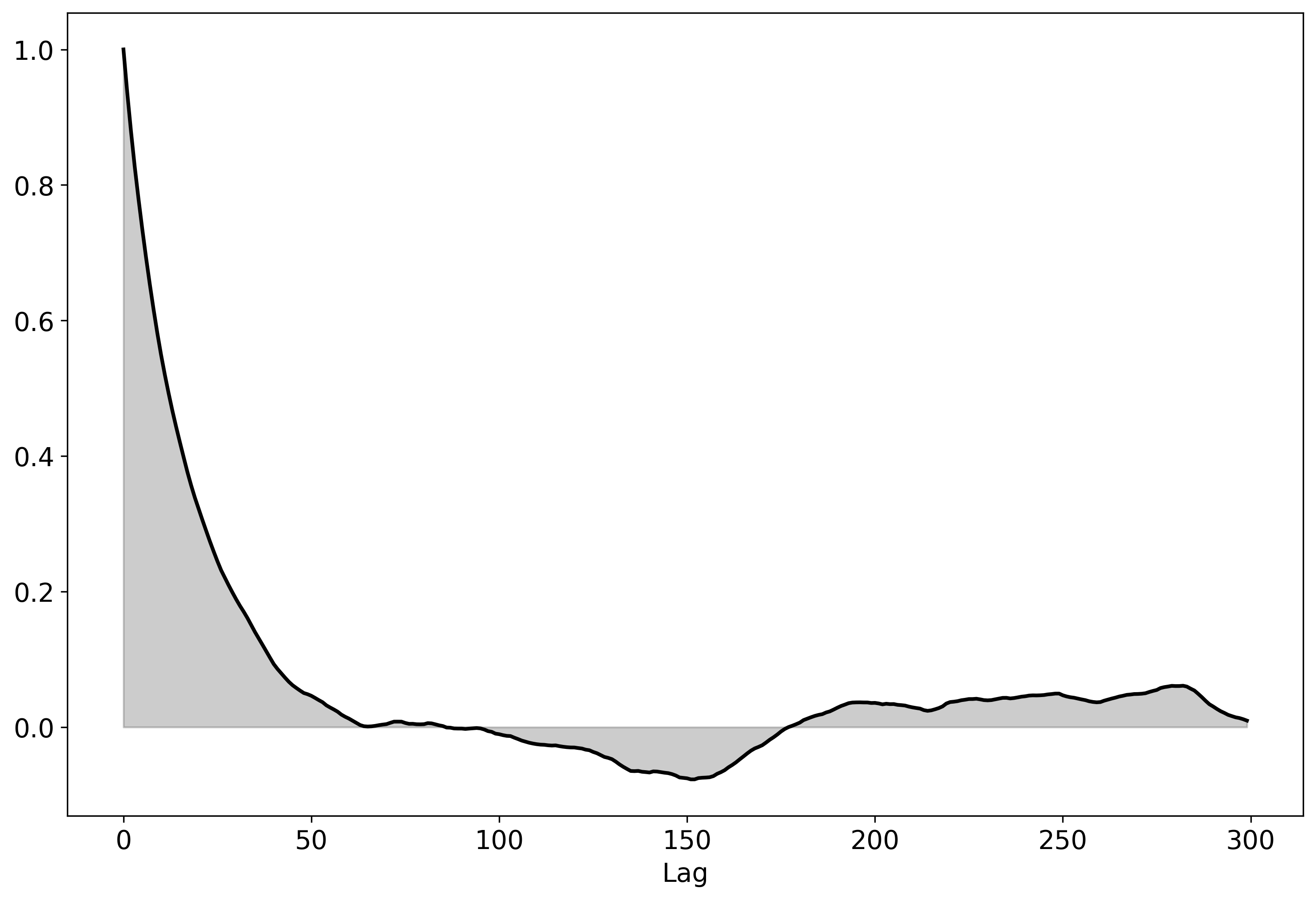}
\caption{Autocorrelation function for $\theta_1$ for samples without AEM (left) and with AEM (right).}
\label{fig:gw_acf}
\end{figure}

To assess the performance of the two approaches, the autocorrelation function (\cref{fig:gw_acf}) and the Effective Sample Size (ESS) for each parameter were computed \cite{vehtari_rank-normalization_2020}. Since the coarsest model was quite a poor approximation of the finest, running MLDA without the Adaptive Error Model (AEM) yielded relatively poor results, with an average ESS of 326 out of 40000 samples, and strong autocorrelation. However, when the AEM was employed and otherwise using the exact same sampling parameters, we obtained an average ESS of 1012 out of 40000 samples, with correspondingly weaker autocorrelation.

Note that this particular numerical experiment was chosen to demonstrate the dramatic effect that employing the AEM can have in MLDA, thus making it possible to use multilevel sampling strategies with very crude approximate models. A FE mesh with 25 degrees of freedom is extremely coarse for a Gaussian random field with correlation length $\lambda=0.1$, yet using the AEM it still provides an excellent surrogate for delayed acceptance. Typically much finer models are used in real applications with longer subchains on the coarser levels (cf.~\cite{dodwell_hierarchical_2015}). The AEM will be less critical in that case and MLDA will also produce good ESS without the AEM.

\section{Conclusions and Future Work}\label{sec:conclusions}

In this paper, we have presented an extension of state-independent Delayed Acceptance MCMC \cite{Chr05}, where a hierarchy of coarse MCMC samplers inform the finest sampler in a cascading fashion. If the models on the coarse levels are carefully designed, the approach can lead to significant computational savings, compared to standard single-level MCMC. A possible direction for future research would be to extend this approach further to the general Delayed Acceptance context, where also state-dependent approximations are supported. We would like to highlight that the choice of proposal on the coarsest level is free, as long as it 
achieves irreducibility for the coarsest distribution. We have chosen relatively simple proposals for the coarsest level, but if e.g. the gradient of the likelihood function is available, one can also employ more advanced gradient-informed proposals, such as MALA, HMC or NUTS.

The presented MLDA algorithm has clear similarities with Multilevel MCMC \cite{dodwell_hierarchical_2015}, in that it allows for any number of coarse levels and extended subchains on the coarse levels, but unlike MLMCMC, it is Markov and asymptotically unbiased, also for finite-length subchains. To achieve this quality, the algorithm must be sequential, which complicates parallelisation considerably. One remedy for this challenge, and a possible direction for future research, would be to employ pre-fetching of proposals \cite{brockwell_parallel_2006}. The central idea of pre-fetching is to precompute proposal ``branches'' and evaluate those in parallel, since for each proposal there are only two options, namely \textit{accept} or \textit{reject}. Pre-fetching and evaluating entire proposal branches is significantly more computationally demanding than the strictly sequential approach and generates more waste, similar to Multiple-Try Metropolis \cite{liu_multiple-try_2000}, since entire branches will effectively be rejected at each step. Minimising the waste of pre-fetching while maintaining the computational gains of parallelisation constitutes a complex, probabilistic optimisation problem. This could be addressed by controlling the pre-fetching length, e.g., using a reinforcement learning agent to learn an optimal policy, and to then hedge bets on valuable pre-fetching lengths, based on the latest sampling history.

A question that remains is the optimal choice of the subchain lengths $\{J_\ell\}_{\ell=1}^L$ for the coarse levels, which is essentially the only tuning parameter in the MLDA algorithm. A good rule of thumb may be to choose the length for any level such that the cost of creating the subchain corresponds to the cost of evaluating a single proposal on the next finer level, but this is not the most rigorous approach. The question has previously been studied in the context of Multilevel Monte Carlo \cite{cliffe_multilevel_2011} and MLMCMC \cite{dodwell_hierarchical_2015}, and involves either computing the optimal (effective) sample size for each level for a fixed acceptable sampling error, or computing the sampling error corresponding to a fixed computational budget. A similar approach can be taken for MLDA, but with some caveats. First, the number of samples on each level is determined, not only by the subchain length on that level, but by the number of samples on the next finer level. Hence, care must be taken when choosing the subchain lengths. Second, it is non-trivial to determine the effective sample size of a level \textit{a priori}, because of the direct correspondence with the distribution on the next finer level by way of the MLDA acceptance criterion. One possible workaround would be to determine the optimal subchain lengths adaptively by empirically determining the effective sample sizes and variances on each level during burn-in. Similarly to the pre-fetching approach outlined above, these decisions could also be outsourced to a reinforcement learning agent that would adaptively learn the optimal policy for minimising either cost or sampling error. We emphasize this question as a potential direction for future research.

\section*{Acknowledgements}
MCMC sampling was completed using the \texttt{MLDA} sampler of the free and open source probabilistic programming library \texttt{PyMC3}. The \texttt{PyMC3} code is available at GitHub: \href{https://github.com/pymc-devs/pymc}{https://github.com/pymc-devs/pymc}. The examples shown in this paper are available at \href{https://github.com/mikkelbue/MLDA\_examples}{https://github.com/mikkelbue/MLDA\_examples}. ML was funded as part of the Water Informatics Science and Engineering Centre for Doctoral Training (WISE CDT) under a grant from the Engineering and Physical Sciences Research Council (EPSRC), grant number EP/L016214/1. TD and GM were funded by a Turing AI Fellowship (2TAFFP\textbackslash100007). CF was partially funded by MBIE contract UOOX2106. The work of RS is supported by the Deutsche Forschungsgemeinschaft (DFG, German Research Foundation) under Germany’s Excellence Strategy EXC 2181/1 - 390900948 (the Heidelberg STRUCTURES Excellence Cluster). The authors have no conflicts of interest to declare.


\bibliographystyle{siamplain}
\bibliography{main} 

\end{document}